\documentclass[prd,nofootinbib,preprintnumbers,preprint
]{revtex4}
\usepackage{empheq}

\usepackage{amsmath}
\usepackage{amsfonts}
\usepackage{amssymb}
\usepackage{dsfont}
\usepackage{bm}
\usepackage[hyperfootnotes=false]{hyperref}
\usepackage[dvipsnames]{xcolor}
\usepackage{apptools}

\newcommand{\Q}{\mathcal Q}
\newcommand{\G}{\mathcal G}
\newcommand{\M}{\mathcal M}

\newcommand{\E}{\mathcal E}
\newcommand{\F}{\mathcal F}
\newcommand{\J}{\mathcal J}

\newcommand{\Sq}{\mathcal{S}_{\mathcal{Q}}}
\newcommand{\LCM}{\nabla}
\newcommand{\LCS}{\mathcal D}
\newcommand{\SFS}{\chi}
\newcommand{\x}{{\bm x}}
\newcommand{\z}{{\bm z}}

\newcommand{\y}{{\bm y}}
\newcommand{\f}{{\bm f}}

\newcommand{\w}{{\bm w}}
\newcommand{\e}{{\bm e}}

\newcommand{\chQ}{{\mathfrak{q}}}
\newcommand{\chP}{{\mathfrak{p}}}
\newcommand{\chSq}{{\mathfrak{e}^2}}

\newcounter{example}[section]
\newcounter{remark}[section]
\newcounter{theorem}[section]
\newcounter{proposition}[section]
\newcounter{lemma}[section]
\newcounter{corollary}[section]
\newcounter{definition}[section]

\AtAppendix{\counterwithin{example}{section}}
\AtAppendix{\counterwithin{remark}{section}}
\AtAppendix{\counterwithin{theorem}{section}}
\AtAppendix{\counterwithin{proposition}{section}}
\AtAppendix{\counterwithin{lemma}{section}}
\AtAppendix{\counterwithin{corollary}{section}}
\AtAppendix{\counterwithin{definition}{section}}

\def\theremark{\arabic{section}.\arabic{remark}}
\def\thetheorem{\arabic{section}.\arabic{theorem}}

\def\thedefinition{\arabic{section}.\arabic{definition}}

\renewcommand*{\email}[1]{\footnote{Electronic address: \href{mailto:#1}{\nolinkurl{#1}} }}

\newenvironment{proof}{\noindent {\textit{Proof:}}
}{\medskip}

\newenvironment{theorem}{\refstepcounter{theorem}
\medskip\noindent{\bf Theorem \thetheorem}:\em}{\medskip}

\newenvironment{definition}{\refstepcounter{definition}\medskip\noindent{\bf
Definition \thedefinition}:\em}{\medskip}

\begin{document}

\title{Glued massive particles surfaces}
\author{Igor Bogush\email{igbogush@gmail.com}}
\author{Kirill Kobialko\email{kobyalkokv@yandex.ru}}
\author{Dmitri Gal'tsov\email{galtsov@phys.msu.ru}}
\affiliation{Faculty of Physics, Moscow State University, 119899, Moscow, Russia}

\begin{abstract}
A novel generalization of photon surfaces  to the case of massive charged particles is given for spacetimes with at least one isometry, including stationary ones. A related notion of {\em glued} massive particle surfaces is also defined. These surfaces join worldlines parametrized by a family of independent conserved quantities and naturally arise in integrable spacetimes. 
We describe the basic geometric properties of such surfaces and their relationship to slice-reducible Killing tensors,  illustrating all concepts with a number of examples. Massive particle surfaces have potential applications in the context of uniqueness theorems, Penrose inequalities, integrability, and the description of black-hole shadows in streams of massive charged particles or photons in a medium with an effective mass and charge. 
\end{abstract}

\maketitle

\setcounter{page}{2}

\setcounter{equation}{0}
\setcounter{subsection}{0}

\section{Introduction}

A new observation window for studying extremely strong gravity near black holes, opened by the  Event Horizon Telescope collaboration \cite{EventHorizonTelescope:2022xnr,EventHorizonTelescope:2019dse}, has stimulated the development of new theoretical tools in the field of gravity, based on the study of compact two-dimensional surfaces in space and the corresponding three-dimensional hypersurfaces in spacetime, in which photons can travel forever \cite{Claudel:2000yi,Koga:2020akc,Cao:2019vlu,Shoom:2017ril,Gibbons:2016isj}. Well-known examples are the photon sphere in the Schwarzschild metric \cite{Virbhadra:1999nm,Virbhadra:2002ju} or spherical surfaces in Kerr spacetime, on which non-equatorial spherical photon orbits wind \cite{Teo:2020sey}. The corresponding hypersurfaces in spacetime can be identified in terms of the geometric theory of submanifolds as three-dimensional umbilical or partially umbilical hypersurfaces \cite{Kobialko:2020vqf,Kobialko:2021uwy}.

The umbilical surface has the proportionality property of the induced metric tensor and the extrinsic curvature tensor, while the partially umbilical condition limits the equality of the first and second fundamental forms not for all tangent vectors, but only for some subspace of the tangent space. In both cases, one can describe such surfaces without resorting directly to the geodesic equations. These timelike surfaces capturing photons have been called fundamental photon surfaces \cite{Kobialko:2020vqf}. Their general properties are interesting not only for understanding the shadows of black holes \cite{Perlick:2021aok,Bronzwaer:2021lzo,GrenzebachSBH,Dokuchaev:2020wqk,Dokuchaev:2019jqq,Cunha:2018acu,Bogush:2022hop,Konoplya:2021slg}, but also as theoretical tools for studying the stability of spacetimes, hidden symmetries \cite{Cariglia:2014ysa,Frolov:2017kze,Gibbons:2011hg} expressed by exact and conformal Killing tensors of the second rank \cite{Kobialko:2021aqg,Kobialko:2022ozq}, Penrose-type inequalities for regions of strong gravity \cite{Shiromizu:2017ego,Feng:2019zzn,Yang:2019zcn,Yoshino:2019dty,Yoshino:2017gqv}, uniqueness theorems \cite{Cederbaum,Yazadjiev:2015hda,Yazadjiev:2015mta,Yazadjiev:2015jza,Yoshino:2016kgi,Yazadjiev:2021nfr,Koga:2020gqd,Rogatko,Cederbaumo,Cederbaum:2019rbv}, etc.  They can exist not only around black holes, but also around wormholes or naked singularities, and they often open a way to distinguish these different ultra-compact objects.

A natural generalization of such characteristic surfaces is to consider timelike orbits of massive particles, as suggested in our previous paper \cite{Kobialko:2022uzj} and in parallel in Ref. \cite{Song:2022fdg} with slightly different perspectives. Note that this also describes the properties of photons in a medium, such as plasma, which may have variable effective mass \cite{Perlick:2017fio} and charge \cite{Asenjo:2010zz,Mendonca:2000tk}. Surfaces that capture massive particles \cite{Teo:2020sey,Hackmann:2013pva} share some common features with photon surfaces, but there are also important differences. The main difference with photons is that massive particles can move at an arbitrary speed, while photons have a fixed speed. For this reason, photon surfaces are more rigid and form a one-parameter family described by their radius only, while in Kerr metric the spherical orbits of massive particles form a {\em two-parameter family}: they can be parameterized by their radius and the value of the Carter constant \cite{Teo:2020sey}. Another important difference is that all spherical orbits of photons in Kerr are unstable, while timelike spherical Kerr geodesics in different parameter ranges can be both unstable and stable. As a result, a direct formal generalization of photon surfaces on the massive particle surfaces is insufficient and requires further ramifications. This is the purpose of this article, in which we introduce a new concept of ``glued'' surfaces, covering their internal splitting in the parameter space.

The article is organized as follows. In Sec. \ref{sec:setup}, we briefly discuss the equations of motion for charged massive particles in spacetimes with several Killing vectors. We also cover the geometric aspects of particle velocity space, conventions of hypersurface geometry, and motivation for further generalization of massive particle surfaces. In Sec. \ref{sec:MPS}, we provide a more general and strict definition of massive particle surfaces, present a key theorem, and discuss the physical and geometric properties of these surfaces. We analyze the resulting geometric constraints, resolve them explicitly, determine the degree of arbitrariness of these surfaces, and discuss surfaces with all Killing vectors tangent to them. In Sec. \ref{sec:GPS}, we define and analyze the effect of gluing of massive particle surfaces, a characteristic feature of integrable systems such as Kerr-like geometries. We examine the phenomenon of maximum gluing and its important geometric consequences. In Sec. \ref{sec:KT} we discuss the close connection between massive particle surfaces and gluing phenomenon in relation to the integrability of dynamical systems and the existence of slice-reducible conformal Killing tensor fields of rank two \cite{Kobialko:2021aqg,Kobialko:2022ozq}. Finally, in Sec. \ref{sec:Examples}, we apply the developed formalism to various important examples, such as the Kerr metric, Zipoy-Voorhees solution, and the dyonic Kerr-Newman solution. The appendix contains proofs of some statements formulated in the main part of the paper. We use symmetrization and antisymmetrization with weight one, e.g., $A_{[\alpha\beta]} = A_{\alpha\beta} - A_{\beta\alpha}$.

\section{Setup} 
\label{sec:setup}

\subsection{Killing vectors}

Let $M$ be a Lorentzian manifold of dimension $n$ with metric tensor $g_{\alpha\beta}$ and Levi-Civita connection $\LCM_\alpha$. We define electromagnetic potential $A_{\alpha}$ and electromagnetic field tensor $F_{\alpha\beta}=\LCM_{[\alpha} A_{\beta]}$ in $M$. The worldline $\gamma^\alpha$ of a test particle with charge $q$ and mass $m$ in this geometry obeys the following equations of motion
\begin{equation} 
   \dot{\gamma}^\alpha \LCM_{\alpha}\dot{\gamma}^\beta =q F^\beta{}_{\lambda}\dot{\gamma}^\lambda , \quad \dot{\gamma}^\alpha\dot{\gamma}_\alpha=-m^2,  \label{eq_particles}
\end{equation} 
where $\dot{\gamma}^\alpha=d \gamma^\alpha/d s$ is a four-velocity of the particle, and $s$ is an affine parameter. One can also use these equations for photons in a medium that can be described by an effective mass $m$ and charge $q$ \cite{Perlick:2017fio,Asenjo:2010zz,Mendonca:2000tk}, which can be a function of coordinates.

Consider a metric $g_{\alpha\beta}$ and an electromagnetic potential $A_\alpha$ that share the same symmetry with respect to Killing vector fields $k_a{}^\alpha$, where the Latin index $a$ enumerates the Killing vectors \cite{Kobialko:2022ozq}. Each Killing vector field satisfies Killing equations $\LCM_{(\alpha}k_a{}_{\beta)}=0$. This implies the existence of conserved quantities $Q_a$ defined as (see Ref. \cite{Kobialko:2022uzj} for details)
\begin{equation}\label{energy}
Q_a\equiv k_{a\alpha}\left(\dot{\gamma}^\alpha+ q A^\alpha\right).
\end{equation}
It is also useful to consider two terms in the expression (\ref{energy}) separately, introducing the ``kinetic'' and ``potential'' components of the conserved quantity $Q_a = q_{a} + q p_{a}$:
\begin{align}
p_a\equiv k_{a\alpha} A^\alpha,\qquad
q_a\equiv Q_a - qp_a.
\end{align}
In the general case, $q_a$ and $p_a$ are not conserved separately. The potential component is a predefined function for given $k_a{}^\alpha$ and $A^\alpha$. The kinetic component $q_a$ can be considered as a secondary quantity which is a certain function of coordinates for fixed conserved quantities $Q_a$. On the other hand, it can be represented as a scalar product of the Killing vector $k_a{}^\alpha$ with some properly normalized timelike (for $m\neq 0$, null for $m=0$) vector $v^\alpha$:
\begin{align} \label{eq:constraint} 
    k_{a\alpha} v^\alpha=q_a=Q_a-qp_a, \quad v_\alpha v^\alpha=-m^2.
\end{align}
Any possible worldline of a particle with fixed $Q_a$, mass $m$, and charge $q$ passing through a given point in spacetime has a 4-velocity coinciding with some vector $v^\alpha$ subject to the constraint (\ref{eq:constraint}).

And vice versa, for any vector $v^\alpha$ obeying constraint (\ref{eq:constraint}) there is a worldline of a particle with fixed $Q_a$, $m$ and $q$, such that its four-velocity coincides with $v^\alpha$. In other words, there is a bijective map between the set of all $v^\alpha$ constrained by (\ref{eq:constraint}) and the set of all charged particle worldlines passing through the given point with fixed $Q_a$, $m$ and $q$ (see details in Ref. \cite{Kobialko:2020vqf} for the massless case).

In what follows, we will assume that the Gram matrix $G_{ab}=k_{a}{}^\beta k_{b}{}_\beta$ of Killing vectors $k_a{}^\alpha$ is non-degenerate\footnote{This condition can be relaxed by using pseudo-inverse matrices. However, this could make the present work unnecessarily technical, so we will not consider the Gram matrix to be pseudo-inverse.}. In particular, all Killing vectors are linearly independent and the inverse matrix $G^{ab}$ is well-defined by the condition $G_{ac}G^{cb}=\delta^a_b$. According to intermediate results (\ref{eq:AB_representation}) and (\ref{eq:solution_A}) of theorem \ref{th:solutions} from App. \ref{sec:appendix_theorems}, the general solution of system (\ref{eq:constraint}) has the form
\begin{align}  \label{eq:constraint_solution} 
 v^\alpha=k_{a}{}^\alpha G^{ab} q_b  +u^\alpha, \quad k_{a}{}_\alpha u^\alpha=0,  \quad u^\alpha u_\alpha=-m^2-G^{ab}q_a q_b. 
\end{align}
Vectors $u^\alpha$ are arbitrary vectors with a fixed norm that are orthogonal to all Killing vectors $k_a{}^\alpha$. The tangent space splits into a direct sum of the subspace generated by the Killing vectors and their orthogonal complement. Since the Gram matrix $G_{ab}$ is not degenerate, these subspaces can have Lorentzian or  Euclidean signatures (but none of them are null). If the orthogonal complement has a Euclidean signature, then it follows from (\ref{eq:constraint_solution}) that the solution exists only if 
\begin{align}  \label{eq:accessible_area} 
G^{ab}q_a q_b\leq -m^2.
\end{align}
If the orthogonal complement has a Lorentzian signature or consists of one timelike vector, there are no restrictions \cite{Kobialko:2022uzj}. In the classical Hamiltonian approach, inequality (\ref{eq:accessible_area}) follows from the effective potential descriptions \cite{Cunha:2018acu}. 

If the orthogonal complement to all Killing vectors has a Euclidean signature, the particle with the given conserved quantities can move only in a certain domain of the spacetime given by inequality (\ref{eq:accessible_area}). In order to describe this domain, we introduce the following definition:

\begin{definition} 
The $\Q$-domain of motion $\M_\Q\subseteq M$ is a maximal subset of $M$ where the system (\ref{eq:constraint}) has a solution.
\end{definition}

From the physical point of view, the motion along the worldlines with given values of the conserved quantities is possible only in a $\Q$-domain of motion $\M_\Q$. A boundary $\partial \M_\Q$ of the $\Q$-domain serves as a set of turning points for the corresponding worldlines. The boundaries correspond to saturation of the inequality (\ref{eq:accessible_area}) if the latter is applicable. Thus, $u^\alpha=0$ and worldlines $\gamma$ are tangent to the Killing vectors at the boundaries. In the particular case $m=0$, it is easy to check that if $\M_\Q$ is connected, it is nothing but the ``causal $\rho$-region'' defined in Ref. \cite{Kobialko:2020vqf} for null geodesics.

\subsection{Hypersurface geometry}

Let $S$ be a timelike hypersurface without boundaries with normal unit vector $n^\alpha$. The main geometric properties of $S$ are specified by the induced metric $h_{\alpha\beta}$ and the second fundamental form $\chi_{\alpha\beta}$:
\begin{equation} \label{eq:fundamental_form}
   h_{\alpha\beta} =
    g_{\alpha\beta}
    -n_\alpha n_\beta, \quad  \SFS_{\alpha\beta}
    \equiv
    h^{\lambda}_{\alpha}
    h^{\rho}_{\beta}
    \LCM_{\lambda} n_{\rho}.
\end{equation}
The Gauss decomposition for two vector fields $v^\alpha$, $u^\alpha$ tangent to $S$ (i.e., $v^\alpha n_\alpha=u^\alpha n_\alpha=0$) reads
\begin{equation} \label{eq:Gauss_decomposition}
    v^\alpha \LCM_\alpha u^\beta =
      v^\alpha \LCS_\alpha u^\beta
    - \SFS_{\alpha\gamma} v^\alpha u^\gamma n^\beta,
\end{equation}
where $\LCS_\alpha$ is  Levi-Civita connection on $S$, defined as  \cite{Kobialko:2022uzj,Kobialko:2022ozq}:
\begin{equation} \label{eq:form.projector_2}
    \LCS_\alpha \mathcal{T}^{\beta\ldots}_{\gamma\ldots} =
    h^{\lambda}_{\alpha} h^{\beta}_{\rho}\ldots  h_{\gamma}^{\tau}\ldots \LCM_\lambda \mathcal{T}^{\rho\ldots}_{\tau\ldots}, 
    \qquad
    \mathcal{T}^{\beta\ldots}_{\gamma\ldots}=h^{\beta}_{\rho}\ldots  h_{\gamma}^{\tau}\ldots T^{\rho\ldots}_{\tau\ldots}.
\end{equation}

An important observation is that the Gauss decomposition (\ref{eq:Gauss_decomposition}) allows us to reduce the problem of analyzing the behavior of worldlines (\ref{eq_particles}) to the analysis of hypersurface fundamental forms (\ref{eq:fundamental_form}). To illustrate this idea, consider a photon surface as an example \cite{Claudel:2000yi}. A photon surface is a hypersurface $S$ in $M$ such that any null geodesic initially tangent to $S$ will remain tangent to it. In other words, any null geodesics tangent to $S$ are entirely in $S$. Applying the Gauss decomposition from Eq. (\ref{eq:Gauss_decomposition}) to the geodesic equation $\dot{\gamma}^\alpha \LCM_\alpha \dot{\gamma}^\beta=0$ for an arbitrary null geodesic $\gamma$, we find 
\begin{equation} 
   0= \dot{\gamma}^\alpha \LCM_\alpha \dot{\gamma}^\beta =
      \dot{\gamma}^\alpha \LCS_\alpha \dot{\gamma}^\beta
    - \SFS_{\alpha\gamma} \dot{\gamma}^\alpha \dot{\gamma}^\gamma n^\beta \quad \Rightarrow \quad\dot{\gamma}^\alpha \LCS_\alpha \dot{\gamma}^\beta=0, \quad \SFS_{\alpha\beta} \dot{\gamma}^\alpha \dot{\gamma}^\beta =0.
\end{equation}
Keeping in mind the correspondence between geodesic and tangent vectors from the previous analysis, we find out that for all null $v^\alpha\in TS$, $v^\alpha h_{\alpha\beta}v^\beta = 0$ the second fundamental form of the photon surface satisfies 
\begin{equation} \label{eq:kkv_1}
    \SFS_{\alpha\beta} v^\alpha v^\beta=0.
\end{equation}
From theorem \ref{th:pre_theorem} in App. \ref{sec:appendix_theorems} follows that Eq. (\ref{eq:kkv_1}) holds for any null $v^\alpha$ if and only if the photon surface $S$ is totally umbilical \cite{Claudel:2000yi}, i.e. its second fundamental form is proportional to the induced metric:
\begin{equation}  \label{eq:totally_umbilical}
\chi_{\alpha\beta} = \frac{\chi}{n-1} h_{\alpha\beta},
\end{equation}
where $\chi={\chi_\alpha}^\alpha$ is an arbitrary scalar function. Condition (\ref{eq:totally_umbilical}) is formulated in terms of the geometric properties of the surface only and does not refer to geodesic equations explicitly. This formulation is an effective way to analyze non-integrable dynamical systems \cite{Frolov:2017kze,Kobialko:2020vqf,Kobialko:2021uwy} and to study general theoretical problems such as Penrose inequalities \cite{Shiromizu:2017ego,Feng:2019zzn,Yang:2019zcn}, uniqueness theorems \cite{Cederbaum,Yazadjiev:2015hda,Yazadjiev:2015mta,Yazadjiev:2015jza,Yoshino:2016kgi,Yazadjiev:2021nfr,Koga:2020gqd,Rogatko,Cederbaumo,Cederbaum:2019rbv} and hidden symmetries \cite{Gibbons:2011hg,Cariglia:2014ysa,Frolov:2017kze,Kobialko:2021aqg,Kobialko:2022ozq}.

In our previous paper \cite{Kobialko:2022uzj} we generalized this geometric result to the case of worldlines of massive electrically charged particles.
To do this, we used explicit spacetime symmetries described by a single Killing vector field.  
This allows us to give a new definition of massive particle surfaces at fixed energy and prove the key theorem about their geometric description. However, as in the case of photon surfaces, the surfaces defined in Ref. \cite{Kobialko:2022uzj} for massive charged particles exist in static or spherically symmetrical spacetimes but not in stationary spacetimes with rotation. Nevertheless, in Ref. \cite{Teo:2020sey} different spherical orbits are shown to fill the whole spheres in Kerr metric. The novel feature is that different worldlines lying in the same sphere have different values of conserved energy. In fact, the energy integral on spherical orbits in Kerr depends on two quantities: their radius and the Carter constant associated with the second-rank Killing tensor. In order to describe this situation in more general terms, here we introduce a generalization of static massive particle surfaces of \cite{Kobialko:2022uzj} in the same fashion as photon surfaces in static spacetimes were generalized to fundamental photon surfaces in the stationary ones in \cite{Kobialko:2020vqf}. 

Recall that fundamental photon surfaces satisfy Eq. (\ref{eq:kkv_1}) only for a subset of null $v^\alpha$ constrained by a linear condition. For Kerr-like metrics, this constraint extracts such vectors $v^\alpha$ that correspond to the fixed impact parameter $\rho = - v_\alpha k_{(\varphi)}^\alpha / v_\beta k_{(t)}^\beta$, where $k^\alpha_{(t)}$, $k^\alpha_{(\varphi)}$ are timelike and azimuthal Killing vectors. First, we will introduce a more general definition of massive particle surfaces that satisfy condition (\ref{eq:kkv_1}) only for $v^\alpha$ corresponding to worldlines with fixed conserved quantities corresponding to a set of Killing vectors. And then we will show that the same hypersurface can correspond not just to one value of conserved quantities, but to some family of conserved quantities with a nonlinear constraint.

\section{Massive particles surfaces}
\label{sec:MPS}

\subsection{Definition and theorem}

The first step toward our goal is to give a more general definition for massive particle surfaces \cite{Kobialko:2022uzj} suitable for spacetimes with a larger number of isometries. Let $S$ be a timelike hypersurface without boundaries. Define projections of Killing vectors $k_a{}^\beta$ as follows
\begin{equation}
 \kappa_a{}^\alpha\equiv h^\alpha_\beta  k_a{}^\beta,
\end{equation}
with a Gram matrix
\begin{equation} \label{eq:GG} 
    \G_{ab} = {\kappa_a}^\alpha{\kappa_b}_\alpha =
    G_{ab} - (k_{a\alpha} n^\alpha)(k_{b\beta} n^\beta).
\end{equation}
As before, we will assume that the matrix $\G_{ab}$ is non-degenerate and the matrix $\G^{ab}$ is its inverse. If some Killing vectors or their projections become linearly dependent or have a singular point, we  exclude them from the set. 

We denote the subspace spanned by projections of Killing vectors as $\mathbf{K}\subseteq TS$ and its orthogonal complement as $\mathbf{K}^\perp$, i.e. $\mathbf{K} \oplus \mathbf{K}^\perp = TS$. Their dimensions are abbreviated as\footnote{In general, the dimensions can change from point to point. But we assume that $n_\perp \geq1$ and $n_K\geq1$.}
\begin{equation}
    n_K \equiv \text{dim}\mathbf{K},\qquad n_\perp \equiv \text{dim}\mathbf{K}^\perp,\qquad
    n_K + n_\perp + 1 = n.
\end{equation}
Using the inverse Gram matrix $\G^{ab}$, one can define projectors to the subspaces $\mathbf{K}$ and $\mathbf{K}^\perp$ as
\begin{equation}
    K^\beta_\alpha = \kappa_{a\alpha} \G^{ab} \kappa_{b}{}^{\beta},\qquad
    U^\beta_\alpha = h_\alpha^\beta - K^\beta_\alpha.
\end{equation}

Note that worldlines with a fixed set of conserved quantities $Q_a$ cannot touch an arbitrary point on the timelike surface $S$.
Indeed, for the worldline to touch $S$, the tangent vectors $v^\alpha$ must satisfy the conditions of equation (\ref{eq:constraint}).
According to the intermediate results (\ref{eq:AB_representation}) and (\ref{eq:solution_A}) of theorem \ref{th:solutions} from Appendix \ref{sec:appendix_theorems}, the general solution of system (\ref{eq:constraint}) with the additional condition $v^\alpha n_\alpha=0$ has the form 
\begin{align}  
 v^\alpha=\kappa_{a}{}^\alpha \G^{ab}q_b  +u^\alpha, \quad n_\alpha u^\alpha=0,  \quad k_{a}{}_\alpha u^\alpha=0,  \quad u^\alpha u_\alpha=-m^2-\G^{ab}q_a q_b.
\end{align}
Thus, as before, if the subspace $\mathbf{K}^\perp$ has a Euclidean signature, then the worldlines can only touch points of the surface $S$ such that
\begin{align}  \label{eq:accessible_surface} 
\G^{ab}q_a q_b\leq -m^2.
\end{align}
If $\mathbf{K}^\perp$ has a Lorentzian signature or consists of a single timelike vector, there are no restrictions. Since it is pointless to look for massive particle surfaces among surfaces that worldlines cannot even touch, we suggest the following:

\begin{definition} 
The $\Q$-touched hypersurface $\Sq \subset S$ is a maximal subset of $S$ where the system (\ref{eq:constraint}) with the additional condition $v^\alpha n_\alpha=0$ has a solution.
\end{definition}
 
As in the case of the $\Q$-domain of motion, the definition of $\Q$-touched hypersurface ensures that worldlines can touch the surface $\Sq$, but can not touch $S\setminus\Sq$.
In particular, the worldline cannot exit $\Sq$ tangentially through its boundary $\partial\Sq$ (consisting of points, where $\G^{ab}q_a q_b=-m^2$). Instead, the worldline can only leave $\Sq$ by moving in the normal direction from the surface. It is clear that any $\Q$-touched hypersurface is inside a $\Q$-domain, i.e. $\Sq\subseteq \M_\Q$. Indeed, by definition, the domain $M\setminus \M_\Q$ does not contain vectors satisfying the constraints (\ref{eq:constraint}). Using Eq. (\ref{eq:GG}), we can see that if all Killing vectors are tangent to the surface, then $q_a \G^{ab} q_b = q_a G^{ab} q_b$. This implies that $\partial \Sq$ is a subset of $\partial \M_\Q$. Now we are ready to select the massive particle surface among the $\Q$-touched hypersurfaces.

\begin{definition} 
A massive particle surface is a $\Q$-touched hypersurface $\Sq$ such that, for every point $p\in \Sq$ and every vector $v^\alpha|_p \in T_p\Sq$ such that $v^\alpha  \kappa_{a\alpha}|_p=q_a$ and $v^\alpha v_\alpha|_p=-m^2$, there exists a worldline $\gamma$ of $M$ for a particle with mass $m$, electric charge $q$ and the integrals of motion $Q_a$ associated with the Killing vectors $k_a{}^\alpha$ such that $\dot{\gamma}^\alpha(0) =v^\alpha|_p$ and $\gamma\subset \Sq$.
\end{definition} 

Simply put, any worldline with a given set of conserved quantities $Q_a$ that touches $\Sq$ (at least at one point) belongs to $\Sq$ entirely. Note that this definition and its implications remain valid for photons in a medium, such as plasma, with the only difference being that the effective photon mass may vary as a function of coordinates \cite{Perlick:2017fio,Asenjo:2010zz}. The key geometric properties of the massive particle surfaces are given in the following theorem: 

\begin{theorem}\label{th:theorem}
Let $\Sq$ be a smooth $q$-touched hypersurface. The following statements are equivalent:

(i) $\Sq$ is a massive particle surface for given $q, m$ and $Q_a$;

(ii) the second fundamental form satisfies 
\begin{equation}\label{eq:condition_ii}
    \SFS_{\alpha\beta} v^\alpha v^\beta=-q n^\beta F_{\beta\lambda} v^\lambda,
\end{equation}
\indent for all $p\in\Sq$ and $\forall v^\alpha \in T\Sq$ such that $v^\alpha v_\alpha=-m^2$ and $v^\alpha \kappa_{a\alpha}=q_a$;

\begin{subequations}\label{eq:item_3}
(iii) the second fundamental form satisfies
\begin{equation}\label{eq:condition_3}
\chi_{\alpha\beta} =
    \chi_\tau \left(
          h_{\alpha\beta}
        + \mathcal{H}^{ab} \kappa_{a\alpha} \kappa_{b\beta}
    \right) 
    + \beta^a{}_{(\alpha} \kappa_{a\beta)},  
\end{equation}
\indent where $\chi_\tau$, $\mathcal{H}^{ab}$ and $\beta^a{}_{\alpha}$ ($ \kappa_{a\alpha} \beta^{b\alpha}=0$) are some functions restricted by the following constraints
\begin{align} \label{eq:condition_q}
     \chi_\tau\left( \mathcal{H}^{ab} q_{a} q_{b} - m^2\right)
    +  q \mathcal{F}^{a} q_a
    = 0,
    \quad
    2 q_a \beta^a{}_{\alpha} = - q f_{\alpha},
\end{align}
\indent and $\mathcal{F}^{a}$, $f_{\alpha}$ represent the following components of the field tensor $F_{\alpha\beta}$
\begin{equation} \label{eq:F_decomposition}
    \F^a = \G^{ab} \kappa_{b}{}^{\beta} n^\alpha F_{\alpha \beta},\quad
    f_{\alpha} = U^\beta_\alpha n^\gamma F_{\gamma \beta},
    \quad\text{i.e.},\quad
    n^\beta F_{\beta\alpha}=\mathcal{F}^{a} \kappa_{a\alpha}+ f_{\alpha},\quad
    \kappa_{a}{}^{\alpha} f_{\alpha}=0.
\end{equation}
\end{subequations}

(iv) every worldline in $\Sq$ with $\dot{\gamma}^\alpha \kappa_{a\alpha}|_p=q_a$ and $\dot{\gamma}^\alpha  \dot{\gamma}_\alpha|_p=-m^2$ at some point $p\in \Sq$ is 
\indent a worldline in $M$. 

\end{theorem} 
\begin{proof}

The proofs that $(i) \Rightarrow (ii)$, $(iii) \Rightarrow (iv)$ and $(iv) \Rightarrow (i)$ are the same as the corresponding proofs from the theorem in Ref. \cite{Kobialko:2022uzj}. Accordingly, we will generalize explicitly only the proof of $(ii) \Rightarrow (iii)$.
We start from the following conditions on vectors $v^\alpha$ tangent to the hypersurface $\Sq$: 
\begin{equation} \label{eq:quantities_for_first_proof_12}
    h_{\alpha\beta} v^\alpha v^\beta=-m^2,\qquad
    v^\alpha \kappa_{a\alpha}=q_a.
\end{equation}
The aim is to answer the question of when are the solutions of this system in the set of solutions of the following system
\begin{equation}
    \SFS_{\alpha\beta} v^\alpha v^\beta=-q n^\beta F_{\beta\lambda} v^\lambda,
\end{equation}
without any additional linear constraints. In order to apply theorem \ref{th:solutions_general} from Appendix \ref{sec:appendix_theorems}, we define the following matrices and vectors:
\begin{subequations}
\label{eq:quantities_for_first_proof}
\begin{align}
    \label{eq:quantities_for_first_proof_1}
    &
    A = h_{\alpha\beta}, \qquad
    B = \kappa_{a\alpha}, \qquad
    \f = q_a, \qquad
    \w = 0,\qquad
    d = -m^2,\qquad
    \x = v^\alpha,
    \\
    \label{eq:quantities_for_first_proof_2}
    &
    A' = \chi_{\alpha\beta}, \qquad
    B' = 0, \qquad
    \f' = 0, \qquad
    \w' = \frac{1}{2}q n^\beta F_{\beta\lambda},\qquad
    d' = 0.
\end{align}
\end{subequations}
All vectors and tensors with Greek indices defined in (\ref{eq:quantities_for_first_proof}) are tangent to $\Sq$ by construction. Therefore, we can work with the matrix $A$ as with the operator $T\Sq \to T^*\Sq$, but not as one acting on the entire spacetime $M$. 
In this case, the operator A is non-degenerate, since it is a non-singular metric in $\Sq$. The non-degeneracy of $A$ is one of the requirements of theorem \ref{th:solutions_general} from Appendix \ref{sec:appendix_theorems}.
The identity operator $\mathds{1}$ in $T\Sq$ corresponds to $h_\alpha^\beta$.
In addition, we must define the matrices $G$, $\Pi$ and $A_\text{pr}$ as follows
(see Eqs. (\ref{eq:G_def}), (\ref{eq:projector_pi}) and (\ref{eq:theorem_condition_5_2}) for details):
\begin{align}
    &\label{eq:system_1_1}
    G = \G^{ab},
    \\\nonumber&
    \Pi = h_\alpha^\beta - \kappa_{a\alpha} \G^{ab} \kappa_{b\gamma} h^{\gamma\beta}
      = U_\alpha^\beta,
      \\\nonumber &\label{eq:system_1_2}
    A_\text{pr} = \Pi^T A \Pi = h_{\alpha\beta} - \kappa_{a\alpha} \G^{ab} \kappa_{b\beta} = U_{\alpha\beta}.
\end{align}
Application of the theorem is possible if the kernel of the matrix $B$ is spanned by solutions $\x$ for system (\ref{eq:system_1_1}). Recall that $\text{ker}\kappa_{a\alpha} = \{x^\alpha:\; \kappa_{a\alpha}x^\alpha = 0\}$, i.e. the kernel of $B$ in this case is the entire subspace of $\mathbf{K}^\perp$. 
This is equivalent to the requirement that projections of the solutions onto the orthogonal complement $u^\alpha = U^\alpha_\beta v^\beta$ span the entire kernel of the matrix $\kappa_{a\alpha}$.
This requirement is valid at interior points of the hypersurface $\Sq/\partial \Sq$, since $u^\alpha$ are all vectors from $\mathbf{K}^\perp$ with fixed norm $u^\alpha u_\alpha= -m^2-\G^{ab}q_a q_b$. The norm sign fits the $\mathbf{K}^\perp$ signature because $\Sq$ is $\Q$-touched. Only on the boundaries $\partial\Sq$ the norm becomes equal to zero, and this condition is violated. We will consider only interior points of $\Sq\setminus\partial \Sq$, but the result can be extrapolated to the entire hypersurface $\Sq$ due to its smoothness.
 
Applying the theorem, we get the form of $\chi_{\alpha\beta}$
\begin{align}
    \chi_{\alpha\beta} &= \chi_\tau \left(h_{\alpha\beta} + \kappa_{a\alpha} \mathcal{H}^{ab} \kappa_{b\beta}\right) + \beta^a{}_{(\alpha} \kappa_{\beta)a},
\end{align}
where we use $\chi_\tau$, $\beta^a{}_\alpha$ and $\chi_\tau \mathcal{H}^{ab}$ instead of an arbitrary function $\lambda$, matrices $S_1$ and $S_2$ from the theorem formulation respectively. The theorem states that the following three conditions must hold
\begin{align}
    &
    0 = \kappa_{a\alpha} \beta^{b\alpha},
    \\\nonumber&
    0 =
    \frac{q}{2}U_\alpha^\beta n^\gamma F_{\gamma\beta} + \beta^a{}_{\alpha} q_a
    =
      \frac{q}{2} f_{\alpha}
    + \beta^a{}_{\alpha} q_a,
    \\\nonumber&
     0 = \chi_\tau\left(q_a \mathcal{H}^{ab} q_b - m^2\right) +q q_a \G^{ab} \kappa_b{}_\alpha h^{\alpha\beta} n^\gamma F_{\gamma\beta}   = \chi_\tau\left(q_a \mathcal{H}^{ab} q_b - m^2\right) + q q_a \mathcal{F}^a,
\end{align}
where the projections $\mathcal{F}^a$, $f_\alpha$ are defined in (\ref{eq:F_decomposition}). This completes the proof of the theorem.
\end{proof}

This theorem is a direct generalization of an analogous result from Ref. \cite{Kobialko:2022uzj} to the case of an arbitrary number of isometries. The basic geometric description of the massive particle surface is given by statement $(iii)$ of the theorem. Let us try to clarify the geometric meaning of this definition. First of all, Eq. (\ref{eq:condition_3}) projected onto $\mathbf{K}$ and $\mathbf{K}^\perp$ gives the identities
\begin{subequations}
    \begin{align}
        \label{eq:chi_components_0.kk}
        \kappa_a{}^\alpha \chi_{\alpha\beta} \kappa_b{}^\beta 
        &=
        \chi_\tau\left(\mathcal{H}_{ab} + \G_{ab}\right),
        \\
        \label{eq:chi_components_0.kt}
        \kappa_a{}^\alpha \chi_{\alpha\beta} U_\lambda^\beta
        &=\G_{ab}\beta^b{}_{\lambda},
        \\
        \label{eq:chi_components_0.tt}
        U_\gamma^\alpha \chi_{\alpha\beta} U_\lambda^\beta &=
        \chi_\tau U_\gamma^\alpha h_{\alpha\beta} U_\lambda^\beta,
    \end{align}
\end{subequations}
where $\mathcal{H}_{ab} = \G_{ac}\mathcal{H}^{cd}\G_{db}$. From Eq. (\ref{eq:chi_components_0.tt}) follows that the massive particle surface is partially umbilical \cite{Kobialko:2020vqf} in $\mathbf{K}^\perp$. A partially umbilical surface is defined by the condition that the extrinsic curvature is the same along certain directions, such as the subspace $\mathbf{K}^\perp$ in our case. When there is only one such direction ($n^\perp = 1$), the condition for partially umbilical surface degenerates, fixing the function $\chi_\tau$. One can extract $\mathcal{H}_{ab}$ and $\beta^a{}_{\alpha}$ from Eqs. (\ref{eq:chi_components_0.kk}) and (\ref{eq:chi_components_0.kt})
\begin{subequations}
\label{eq:chi_components}
\begin{align}
    \label{eq:chi_components.kk}
    \mathcal{H}_{ab} &=
             \frac{1}{\chi_\tau}\kappa_a{}^\alpha \kappa_b{}^\beta \chi_{\alpha\beta} - \G_{ab},\\\label{eq:chi_components.kt}
   \beta^a{}_{\lambda}&=G^{ab} \kappa_b{}^\alpha U_\lambda^\beta \chi_{\alpha\beta},
\end{align}
\end{subequations}
where we assume that $\chi_\tau$ is not zero. 

\subsection{Extended massive particle surface}

It is convenient to ignore condition (\ref{eq:accessible_surface}) and consider the hypersurface at first sight without it. To do that, we introduce an \textit{extended massive particle surface} $\mathcal{S}_e$. Namely, the massive particle surfaces $\Sq$ are defined as a part of some surface $S$. The geometry of the complement part of the surface $S\setminus\Sq$ is not fixed in any way, since none of the worldlines with the corresponding $Q_a$ can touch these points. The most natural extension of $\Sq$ seems to be an entirely smooth one without boundaries. We require $S=\mathcal{S}_e$ to possess the second fundamental form satisfying Eq. (\ref{eq:condition_3}) with conditions (\ref{eq:condition_q}) in all its points. Of course, in the general case, such an extension may not exist. Altogether, the definition of the extended massive particle surface is the following

\begin{definition}
An extended massive particle surface is a hypersurface $\mathcal{S}_e$ with no boundaries, such that for every point $p\in \mathcal{S}_e$, the second fundamental form takes the form (\ref{eq:condition_3}) subject to the imposed conditions (\ref{eq:condition_q}).
\end{definition}

Obviously, any $\Q$-touched part of the extended massive particle surface is an ordinary massive particle surface. However, $\mathcal{S}_e$ has a stiffer geometry in the complement part $S\setminus\Sq$. In our opinion, this restriction will be extremely useful for the analysis of both ordinary surfaces of massive particles and general issues such as integrability and Penrose inequalities. In particular, boundaryless $\mathcal{S}_e$ surfaces require less technical detail and suggest more topological applications, making them a particularly useful object of study.
 
\subsection{Constraints}

Conditions (\ref{eq:condition_q}) are invariant under the set of transformations 
\begin{align}
q_a\rightarrow \lambda q_a, \quad q\rightarrow \lambda q, \quad m^2\rightarrow \lambda m^2.
\end{align}
In particular, for the null geodesics, there is an invariance $q_a\rightarrow \lambda q_a$ associated with the conformal invariance of photon worldlines. Note that, any set $Q_a$ can have its own unique hypersurface. However, as we will see later, the hidden symmetries of space can lead to gluing different surfaces together and forming complete hypersurface $S$ without boundaries. 

Conditions (\ref{eq:condition_q}) impose some restrictions on the components of the second fundamental form $\chi_{\alpha\beta}$. However, not all components are generally rigidly fixed, and some of them can be completely arbitrary. To explicitly reveal the role of constraints (\ref{eq:condition_q}), we define the vector $\tau^\alpha$ (similarly to the covector $\tau_\alpha$, the dual vector of the impact parameter $ \rho^\alpha$ in Ref. \cite{Kobialko:2020vqf}) in the subspace $\mathbf{K}$ as
\begin{align}\label{eq:rho_def}
  \tau^\alpha=h^{\alpha \beta} \kappa_{a\beta} \G^{ab} q_b, \quad \tau^2\equiv \tau^\alpha  \tau_\alpha=q_a\G^{ab}q_b.
\end{align}

If $\mathbf{K}$ is Euclidean, then $\tau^2$ is strictly positive for any nonzero $q_a$. Otherwise, it is non-zero due to the inequality (\ref{eq:accessible_area}) for $m \neq 0$. Then we define the projection $T^\alpha_\beta$ along the direction $\tau^\alpha$ and the projector $P^\alpha_\beta$ onto the orthogonal complement $\tau^\alpha$ in $\mathbf{K}$:
\begin{align}\label{eq:rho_projectors}
  T_\beta^\alpha= \tau^{-2} \tau^\alpha \tau_\beta, \quad P_\beta^\alpha=K_\beta^\alpha - T_\beta^\alpha. 
\end{align}
The second and third terms in the equation (\ref{eq:condition_3}) can be further projected using (\ref{eq:rho_projectors}) as follows:
\begin{align} \label{eq:chi_rp}
    \SFS_{\alpha\beta} = 
    \chi_\tau h_{\alpha\beta}
    + z \tau_\alpha \tau_\beta
    + \sigma_{\alpha\beta}
    + \tau_{(\alpha} z_{\beta)},
\end{align}
where $z^\alpha$ and $\sigma_{\alpha\beta}=\sigma_{\beta\alpha}$ have the following projection properties
\begin{align} \label{eq:sigma_def}
    K_{\alpha'}^\alpha z_{\alpha}=0,
    \qquad
    T_{\alpha'}^\alpha \sigma_{\alpha\beta} T_{\beta'}^\beta = 0,
    \qquad
    T_{\alpha'}^\alpha \sigma_{\alpha \beta} U_{\beta'}^\beta = 0,
    \qquad
    U_{\alpha'}^\alpha \sigma_{\alpha \beta} U_{\beta'}^\beta = 0.
\end{align}
Plugging Eqs. (\ref{eq:chi_components}), (\ref{eq:chi_rp}) and the identity $q_a = \kappa_{a\alpha}\tau^\alpha$ in the constraints (\ref{eq:condition_q}), we get the following independent of $\sigma_{\alpha\beta}$ results
\begin{equation}
    z = \tau^{-4} \left(\chi_\tau m^2-q \mathcal{F}^{a} \kappa_{a\alpha}\tau^\alpha\right),
    \qquad
    z_{\alpha} = - \frac{1}{2} \tau^{-2} q f_{\alpha}.
\end{equation}

Finally, we can reformulate statement \textit{(iii)} in theorem \ref{th:theorem} without additional restrictions as follows:

{\it
(iii) the second fundamental form can be represented as
\begin{align} \label{eq:sigma}
    \SFS_{\alpha\beta} = 
    \chi_\tau h_{\alpha\beta}
    + \tau^{-4} \left(\chi_\tau m^2-q \mathcal{F}^{a} \kappa_{a\gamma}\tau^\gamma\right) \tau_\alpha \tau_\beta
    + \sigma_{\alpha\beta}
    - \frac{1}{2} \tau^{-2} q \tau_{(\alpha} f_{\beta)},
\end{align}
where $\chi_\tau$ is some function, the vector $\tau^\alpha$ is defined in equation (\ref{eq:rho_def}), an arbitrary symmetric matrix $\sigma_{\alpha\beta}$ is constrained by the conditions Eq. (\ref{eq:sigma_def}), while $\mathcal{F}^{a}$, $f_{\alpha}$ represent the components of the $F_{\alpha\beta}$ field tensor defined in (\ref{eq:F_decomposition}).
}

The number of unrestricted components in $\chi_{\alpha\beta}$ is placed in $\sigma_{\alpha\beta}$. In total $\sigma_{\alpha\beta}$ has $n(n-1)/2$ components restricted by projectors.  There are $1$, $n_\perp$, and $n_\perp(n_\perp + 1)/2$ independent constraints in each condition on $\sigma_{\alpha\beta}$ from (\ref{eq:sigma_def}) respectively. 
By subtracting the number of constraints from the total number of components in $\sigma_{\alpha\beta}$, we get $(n_K-1)(2n-n_K)/2$ unrestricted components. In the case of $n_K=1$ there are no unrestricted components at all, and the outer geometry of the hypersurface is uniquely determined by the inner geometry, as expected \cite{Kobialko:2022uzj}. For null geodesics, we simply get
\begin{equation}\label{eq:condition_3_*} 
\SFS_{\alpha\beta}=\chi_\tau h_{\alpha\beta}+\sigma_{\alpha\beta}.
\end{equation}
Thus, the hypersurface is umbilical for subspace $\{T^\alpha_\beta, U^\alpha_\beta\}$ as expected for fundamental photon surfaces \cite{Kobialko:2020vqf}.

\subsection{One Killing vector}
In the case of only one Killing vector field $k^\alpha$, all components of the second fundamental form turn out to be strictly fixed (in this subsection, we omit the indices $a$, $b$, since there is only one direction). Indeed, equation (\ref{eq:condition_q}) implies ($\E_k=-q_k$)
\begin{align}
      \mathcal{H} 
    = \frac{m^2}{\E^2_k}+\frac{q}{\chi_\tau \E_k} \mathcal{F},
    \quad
    \beta_{\alpha} =\frac{q}{2\E_k} f_{\alpha},
\end{align}
Substituting this back to Eq. (\ref{eq:condition_3}), the second fundamental form is read 
\begin{align}
    \chi_{\alpha\beta} = \chi_\tau \left(
          h_{\alpha\beta}
        + \frac{m^2}{\E^2_k} \kappa_{\alpha} \kappa_{\beta}
    \right)+\frac{q}{2\E_k}  n^\rho F_{\rho(\alpha} \kappa_{\beta)}.
\end{align}
This is exactly the same property \textit{(iii)} from the theorem in Ref. \cite{Kobialko:2022uzj}. Thus, the new massive particle surfaces are in fact a direct generalization of similar surfaces from Ref. \cite{Kobialko:2022uzj} for spacetimes admitting a larger number of isometries.

\subsection{Orthogonal Killing vector}
If there exists a Killing vector $k_{\perp}^\alpha$ that is everywhere orthogonal to the hypersurface $\Sq$, then $\Sq$ is a totally geodesic hypersurface, i.e., $\SFS_{\alpha\beta}=0 $. This follows from the Killing equation $\LCM_{(\alpha}k_{\perp}{}_{\beta)}=0$ and the symmetry of the product $h^{\lambda}_{\alpha}h^{\rho}_{\beta}$:
\begin{equation}
    \SFS_{\alpha\beta}
    \equiv
    h^{\lambda}_{\alpha}
    h^{\rho}_{\beta}
    \LCM_{\lambda} n_{\rho}
    =
    h^{\lambda}_{\alpha}
    h^{\rho}_{\beta}
    \LCM_{\lambda} (\mathfrak{K} k_{\perp}{}_{\rho})
    =
    h^{\lambda}_{\alpha}
    h^{\rho}_{\beta}
    \left[
          \mathfrak{K} \LCM_{\lambda} k_{\perp}{}_{\rho}
        + k_{\perp}{}_{\rho} \LCM_{\lambda} \mathfrak{K}
    \right]
    =
    0,
\end{equation}
where the normal vector is expressed through the Killing vector $n^\alpha = \mathfrak{K} k^\alpha_\perp$.
The worldline $\gamma$ can touch the hypersurface $\Sq$ if and only if $q_\perp = 0$.
Since the condition $\SFS_{\alpha\beta}=0$ results in $\chi_\tau=0$, $\beta_a{}^\alpha=0$, constraints (\ref{eq:condition_q}) can be simplified to the conditions for the Lorentz force acting on particles lying on the surface 
\begin{align}
    q \mathcal{F}^{a} q_a = 0,
    \quad
    q f_{\alpha} = 0,
\end{align}
i.e., either particles are neutral $q=0$, or $\mathcal{F}^{a} q_a = 0$, $f_\alpha = 0$. 

\subsection{Tangent Killing vectors}

Of great interest is the class  of massive particle surfaces that have the same set of symmetries as the original spacetime. It is well known that surfaces are invariant under the action of some isometry group if the corresponding Killing vector field is tangent to $S$. Thus, we assume that the Killing vectors touch the massive particle surface  (not necessarily all the spacetime Killing vectors, but at least those that span the $\mathbf{K}$ subspace), and they are all linearly independent. If they are not linearly independent, remove the redundant vectors from the considered set of Killing vectors. The mixed components of the second fundamental form can be rewritten as
\begin{align}
    U_\gamma^{\alpha} \kappa^\beta_a \chi_{\alpha\beta} =
    U_\gamma^{\alpha} \kappa^\beta_a \nabla_\alpha n_\beta =
    -U_\gamma^{\alpha} n^\lambda \nabla_\alpha \kappa_{a\lambda}.
\end{align}
Combining with Eq. (\ref{eq:chi_components.kt}), we get
\begin{align}\label{eq:tangent_beta}
    \beta^a{}_{\alpha}
    = - \G^{ab} n^\lambda U_{\alpha}^{\beta} \nabla_\beta \kappa_{b\lambda}.
\end{align}
Substituting this expression into the linear constraint in Eq. (\ref{eq:condition_q}), we get a constraint on some components of $F_{\alpha\beta}$
\begin{align} 
    U^\beta_\gamma n^\alpha \left(
          q F_{\alpha\beta}
        - 2 q_a \G^{ab} \nabla_\beta \kappa_{b\alpha}
    \right) = 0,
\end{align}
Before proceeding to the analysis of the components of the second fundamental form tangent to the Killing vectors, we derive several identities. From the Killing equations, it follows
\begin{align} \label{eq:tangent_k_1}
    n^\beta (\kappa_a{}^\alpha \nabla_\alpha \kappa_{b\beta} + \kappa_b{}^\alpha \nabla_\alpha \kappa_{a\beta}) = 
    - n^\beta (\kappa_a{}^\alpha \nabla_\beta \kappa_{b\alpha} + \kappa_b{}^\alpha \nabla_\beta \kappa_{a\alpha}) &= \\\nonumber
   = - n^\beta \nabla_\beta (\kappa_a{}^\alpha \kappa_{b\alpha})& =
    - n^\beta \nabla_\beta \G_{ab}. 
\end{align}
Using the symmetry of the second fundamental form ($\chi_{\alpha\beta}=\chi_{\beta\alpha}$), we get the expression for the components of the second fundamental form tangent to the Killing vectors
\begin{align} \label{eq:tangent_k_2}
    &
    \kappa_a{}^\alpha \kappa_b{}^\beta \chi_{\alpha\beta} = 
    \kappa_a{}^\alpha \kappa_b{}^\beta \nabla_\alpha n_\beta = 
    - n^\beta \kappa_a{}^\alpha \nabla_\alpha \kappa_{b\beta} =
    - \frac{1}{2} n^\beta \left(
          \kappa_a{}^\alpha \nabla_\alpha \kappa_{b\beta}
        + \kappa_b{}^\alpha \nabla_\alpha \kappa_{a\beta}
    \right) =
    \frac{1}{2} n^\beta \nabla_\beta \G_{ab}.
\end{align}
Comparing Eq. (\ref{eq:tangent_k_2}) and (\ref{eq:chi_components.kk}), we get the expression for the tensor $\mathcal{H}_{ab}$
\begin{align}\label{eq:HasG}
    \mathcal{H}_{ab} = \frac{1}{2\chi_\tau} n^\beta \nabla_\beta \G_{ab} - \G_{ab}
    \qquad\text{or}\qquad
    \mathcal{H}^{ab} =- \frac{1}{2\chi_\tau} n^\beta  \nabla_\beta\G^{ab} - \G^{ab},
\end{align}
Quadratic constraint (\ref{eq:condition_q}) becomes the master equation \cite{Kobialko:2021uwy,Kobialko:2021aqg} for massive particle surface
\begin{align} 
    \frac{1}{2} n^\beta  \nabla_\beta\G^{ab} q_{a} q_{b} + \chi_\tau\left(\G^{ab} q_{a} q_{b} + m^2\right)
    -  q \mathcal{F}^{a} q_a
    = 0.
\end{align}
Thus the second fundamental form of such a massive particle surface has no arbitrary components, i.e. tensor $\sigma_{\alpha\beta}$ in expression (\ref{eq:sigma}) turns out to be completely fixed. 

\section{Glued particle surfaces}
\label{sec:GPS}

The surfaces defined in the previous section usually capture worldlines only for certain values of the integrals of motion.
However, in a Kerr-type spacetime that allows spherical orbits of constant radius, each sphere captures worldlines defined by a one-parameter family of integrals of motion \cite{Teo:2020sey}. Therefore we find it useful to introduce the ``glued'' massive particle surfaces as follows:

\begin{definition} 
A glued particle surface $\mathcal{S}_g$ of the order $g$ is a hypersurface, which is an extended massive particle surface for a family of $Q_a$ of dimension $g \geq 1$ for fixed mass $m$ and electric charge $q$.
\end{definition}

Less formally, a glued particles surface $\mathcal{S}_g$ is a continuous stack of extended massive particle surfaces, representing the same hypersurface. Obviously, $\mathcal{S}_g$ contains all the corresponding massive particle surfaces (not extended), and  the order of its gluing is the same for all points of $p\in\mathcal{S}_g$. 

To find out whether the extended massive particle surface is glued, we calculate the variation of constraints (\ref{eq:condition_q}) with respect to $Q_a$
\begin{align}
    \left(
        2 \chi_\tau\mathcal{H}^{ab} q_{b}
        + q \mathcal{F}^{a}
    \right) \delta Q_a
    = 0,
    \quad
    \beta^a{}_{\alpha} \delta Q_a = 0.
\end{align}
The massive particle surface is glued if there exists a nontrivial solution for the constant $\delta Q_a$ on the surface. The dimension of the space of constant solutions $\delta Q_a$ determines the order of the glued surface.
 
We can also approach the problem from a different angle. Namely, consider some timelike surface $S$ without boundaries, partially umbilical everywhere, i.e. the second fundamental form takes form (\ref{eq:condition_3}). At each point of such a surface, one can calculate the matrices $\mathcal{H}^{ab}$ and $\beta^a{}_{\alpha}$ and consider the conditions from the equation (\ref{eq:condition_q}) as a system of quadratic and linear equations in unknown real variables $Q_a$. Clearly, the dimension of this family determines the possible \textit{maximal gluing order} $g_\text{max}$. If the hyperplane in space of $Q_a$ determined by the linear constraint intersects the quadric hypersurface determined by the quadratic equation, the maximal gluing order is $g_\text{max} = n_K - 1 - \text{rank}B$. If the hyperplane is only tangent to the hypersurface along some directions, then $g_\text{max} < n_K - 1 - \text{rank}B$. 

The solution $Q_a$ corresponds to the extended massive particle surface $\mathcal{S}_e$ if and only if it satisfies condition (\ref{eq:condition_q}) at all points $p\in\mathcal{S}_e $. Thus, we can apply any Lie derivative to condition (\ref{eq:condition_q}) while keeping $Q_a$ unchanged. The dimension of the subset of such solutions $Q_a$, satisfying the conditions on the entire extended massive particle surface, determines the actual order of gluing $g\leq g_\text{max}$. If the number of such solutions is countable, then the surface is not glued.

We rewrite system (\ref{eq:condition_q}) in the form given in Eq. (\ref{eq:main_theorem_system_1}) from theorem \ref{th:solutions_general} as follows
\begin{align} \label{eq:Q_system}
    &
      Q_{a} \mathcal{H}^{ab} Q_{b}
    + 2 \J^a Q_a
    =
    \mathcal{M}^2,\qquad
    \beta^a{}_{\alpha} Q_a = j_\alpha,
\end{align}
where we introduce the following notations
\begin{subequations}
\begin{align} \label{eq:jm}
    \J^a &= 
    q\left(
          \frac{1}{2\chi_\tau} \mathcal{F}^{a} 
        - \mathcal{H}^{ab} p_{b}
    \right)
    ,\qquad  j_\alpha=q(p_a \beta^a{}_{\alpha} - f_\alpha/2), \\
     \mathcal{M}^2 &= m^2
    + q^2 p_a \left(
          \frac{1}{\chi_\tau} \mathcal{F}^{a}
        - \mathcal{H}^{ab} p_{b}
    \right).
\end{align}
\end{subequations}
The transition to matrix notation in theorem \ref{th:solutions_general} is achieved by the following identification
\begin{equation}
    \x=Q_a, \qquad A = \mathcal{H}^{ab}, \qquad B = \beta^a{}_\alpha, \qquad \w = \J^a,\qquad
    d = \mathcal{M}^2,\qquad
    \f = j_\alpha,  
\end{equation}
The primed system (\ref{eq:main_theorem_system_2}) can be represented as the Lie derivative of condition (\ref{eq:Q_system}):
\begin{equation}
    A' = \LCS_\gamma \mathcal{H}^{ab}, \qquad
    B' = \LCS_\gamma \beta^a{}_\alpha, \qquad
    \w' = \LCS_\gamma \J^a,\qquad
    d' = \LCS_\gamma \mathcal{M}^2,\qquad
    \f' = \LCS_\gamma j_\alpha.
\end{equation}

For maximum gluing $g = g_\text{max}$, any solution of system (\ref{eq:Q_system}) must also satisfy the primed system, and we can apply theorem \ref{th:solutions_general}, which will impose additional conditions on the surface geometry. A complete general analysis of an arbitrary glued particle surface using theorem \ref{th:solutions_general} is voluminous, since it depends on various factors, such as the kernels of the $A$ and $B$ matrices, their intersection, the relation between the vectors $\mathcal{J} ^a$ and these kernels, the $\mathcal{M}^2$ sign, etc. Therefore, in this paper we will consider only the simple case $\beta^a{}_\alpha = 0$ to demonstrate the applicability of the formalism, deferring a more complete analysis to future work.
~

\textbf{Example: Case $\beta^a{}_\alpha = 0$.}
The case of maximum gluing $g = g_\text{max}=n_K-1$ without the linear constraint $\beta^a{}_\alpha = 0$ offers a simpler analysis but yields fruitful results. First of all, we obtain that $f_\alpha = 0$, $B=B'=0$, $\f = \f'= 0$, and the pseudoinverse matrix $G$ can be set to zero. This follows from theorem \ref{th:solutions_general} that
\begin{align}
    A' &= \lambda A,\qquad
    \w' = \lambda \w,\qquad
     d' = \lambda d
\end{align}
without any additional constraints. Note that an arbitrary function $\lambda$ has a vector index because the primed quantities contain the derivative $\LCS_\tau$ along any direction tangent to the hypersurface. Hence, by expressing these equations in terms of the derivative, we obtain
\begin{subequations}\label{eq:Dcondition_all}
\begin{align}
    & \label{eq:Dcondition_a}
    \LCS_\tau \mathcal{H}^{ab} = \lambda_\tau \mathcal{H}^{ab},
    \\ & \label{eq:Dcondition_b}
    \LCS_\tau \J^a = \lambda_\tau \J^a,
    \\ & \label{eq:Dcondition_c}
    \LCS_\tau \mathcal{M}^2 = \lambda_\tau \mathcal{M}^2.
\end{align}
\end{subequations}
In the chargeless massive case, Eq. (\ref{eq:Dcondition_b}) holds automatically, from Eq. (\ref{eq:Dcondition_c}) follows $\lambda_\tau = 0$, and the only condition is $\LCS_\tau \mathcal{H}^{ab} = 0$. In the null case $q=m=0$, Eqs. (\ref{eq:Dcondition_b}) and (\ref{eq:Dcondition_c}) holds automatically, and the only condition is (\ref{eq:Dcondition_a}). It is worth mentioning with no detail that one can similarly apply theorem \ref{th:solutions_general} to the neutral case $q=0$ with arbitrary $\beta^a{}_\alpha$, getting the same result that either $\lambda_\tau = 0$ or $m = 0$. One can find from Eq. (\ref{eq:jm}) that ${\mathcal{M}^2 = m^2 + q p_a \left( \J^a + q\mathcal{F}^{a} / 2\chi_\tau \right)}$, which can be used to combine Eqs. (\ref{eq:Dcondition_all}) into new ones (we assume that the mass $m$ is constant)
\begin{subequations}\label{eq:Dcondition_alter}
\begin{align} \label{eq:Dcondition_b_a}
  q\mathcal{H}^{ab} \LCS_\tau  p_{b} &= q\LCS_\tau\left(
          \frac{1}{2\chi_\tau} \mathcal{F}^{a} \right) 
    - q\lambda_\tau\left(
          \frac{1}{2\chi_\tau} \mathcal{F}^{a} 
    \right), \\
q^2 \mathcal{F}^{a} \LCS_\tau p_a &= \lambda_\tau \chi_\tau m^2, \label{eq:Dcondition_c_a}
\end{align}
\end{subequations}
which can replace Eq. (\ref{eq:Dcondition_b}) and (\ref{eq:Dcondition_c}).

Generally, Eqs. (\ref{eq:Dcondition_all}) can be represented as $\LCS_\tau X = \lambda_\tau X$, where $X\in\{\mathcal{H}^{ab}, \J^a, \mathcal{M}^2\}$ is a scalar with respect to connection $\LCS_\tau$. Let us calculate the commutator $[\LCS_\tau, \LCS_\sigma]$ acting on any of these $X$'s. Since it is a scalar, the commutator must be zero $[\LCS_\tau, \LCS_\sigma] X = 0$. On the other hand, from Eqs. (\ref{eq:Dcondition_all}) follows
\begin{equation}
    0 = [\LCS_\tau, \LCS_\sigma] X = (\LCS_\tau \lambda_\sigma - \LCS_\sigma \lambda_\tau)X
    \qquad\Rightarrow\qquad 
    \LCS_\tau \lambda_\sigma - \LCS_\sigma \lambda_\tau = 0,
\end{equation}
which ensures that $\lambda_\tau$ can be expressed locally as 
\begin{equation}\label{eq:lambda_potential}
    \lambda_\tau = -\LCS_\tau \ln \Lambda,
\end{equation}
where $\Lambda$ is some positive function. In this case, Eqs. (\ref{eq:Dcondition_all}) can be represented as
\begin{subequations}\label{eq:Dcondition_all2}
\begin{align}
    & \label{eq:Dcondition_a2}
    \LCS_\tau (\Lambda \mathcal{H}^{ab}) = 0,
    \\ & \label{eq:Dcondition_b2}
    \LCS_\tau (\Lambda \J^a) = 0,
    \\ & \label{eq:Dcondition_c2}
    \LCS_\tau (\Lambda \mathcal{M}^2) = 0.
\end{align}
\end{subequations}
From Eq.(\ref{eq:Dcondition_c2}) follows that one can set $\Lambda=M^{-2}$ up to some multiplicative function that is constant on $\mathcal{S}_g$. This leaves us with only two first equations
\begin{align} \label{HJasC}
     \mathcal{H}^{ab} = M^{2} C^{ab},\quad
    \J^a =M^{2} C^a,
\end{align}
where $C^{ab},\,C^{a}$ are constant in $\mathcal{S}_g$. Eqs. (\ref{HJasC}) mean that the original quadratic equation (\ref{eq:Q_system}) is constant on the surface up to some arbitrary multiplicative function.
In the final section of the article, we show by a number of examples that in Kerr-like geometries all the discussed gluing conditions are indeed satisfied. But first, let's try to understand the deeper cause of the gluing phenomenon and its consequences.
 
\section{Relation to Killing tensors}
\label{sec:KT}

The existence of glued surfaces of massive particles in the Kerr metric is closely related to the existence of hidden spacetime symmetries described by the Killing tensor of the second rank.
For many physically relevant spacetimes, the Killing tensor of the second rank is slice-reducible with respect to some foliation, so that the Killing vectors are slice-tangent, $k_a^\alpha = \kappa_a^\alpha$ (see \cite{Kobialko:2022ozq} for details). According to ref. \cite{Kobialko:2022ozq}, the slice-reducible Killing tensor of the second rank (with not quite geodesic slices) has the form  
\begin{equation}\label{eq:sr_killing}
    \mathcal{K}_{\alpha\beta} = \alpha g_{\alpha\beta} + \gamma^{ab}\kappa_a{}_\alpha \kappa_b{}_\beta + e^\Psi n_\alpha n_\beta,
\end{equation}
where $\alpha$, $\gamma^{ab}$, $\Psi$ are functions obeying some differential equations presented in  Ref. \cite{Kobialko:2022ozq}. In particular, the functions $\alpha$ and $\gamma^{ab}$ are constant along hypersurfaces, i.e., $\LCS_\tau \alpha = 0$ and $\LCS_\tau \gamma^{ab} = 0$. 

A number of restrictions are also automatically applied to slices. In particular, they must be partially umbilical, i.e.,  the second fundamental form provides conditions (\ref{eq:condition_3}) and the mixed components are zero $\beta^a{}_\alpha = 0$ (see Eq. (21) in \cite{Kobialko:2022ozq}\footnote{The symbol $\beta^a_\alpha$ corresponds to $\mathcal{X}_i^\beta$ in the reference.}). Upon comparing Eq. (\ref{eq:HasG}) presented in this work and Eq. (34b) in Ref. \cite{Kobialko:2022ozq}, it is evident that in slices $\LCS_\tau \mathcal{H}^{ab} = \lambda_\tau \mathcal{H}^{ab}$ where $\lambda_\gamma = - \LCS_\gamma\ln\left(\chi_\tau\varphi^3\right)$ (or $\Lambda=\chi_\tau \varphi^3$ up to a multiplicative constant). 
Here $\varphi$ is the lapse function of the foliation defined as\footnote{In original equations from Ref. \cite{Kobialko:2022ozq} there is an additional $\epsilon$ which is equal to 1 for timelike hypersurfaces.}
\begin{equation}
    n^\alpha \nabla_\alpha n_\beta = -\LCS_\beta \ln \varphi.
\end{equation}
Thus, we can establish a connection between slice-reducible conformal Killing tensor and massive particle surfaces \cite{Kobialko:2022ozq}. In order to find out whether the slice is a massive particle surface, one can differentiate Eq. (\ref{eq:Q_system}) with respect to any direction tangent to the hypersurface and combine it with condition (\ref{eq:Dcondition_a}), which holds automatically for the slice. As a result, we obtain only the following remaining condition
\begin{align} \label{eq:jshy}
2  (\LCS_\tau\J^a -  \lambda_\tau \J^a) Q_a=\LCS_\tau \mathcal{M}^2 - \lambda_\tau   \mathcal{M}^2, 
\end{align}
for some $Q_a$ satisfying (\ref{eq:Q_system}). In particular, the maximum gluing order is achieved under conditions (\ref{eq:Dcondition_b}) and (\ref{eq:Dcondition_c}) or equivalently (\ref{eq:Dcondition_alter}) which resolve the condition (\ref{eq:jshy}). 

\begin{theorem}
Let a timelike foliation of the manifold $M$ generates slice-reducible conformal Killing tensor, i.e., satisfies all conditions of the theorem 3.5 in Ref. \cite{Kobialko:2022ozq}. If a slice $S$ has no boundary and satisfies (\ref{eq:Dcondition_alter}), it is an extended glued particle surface $\mathcal{S}_g$ of maximal order.
\end{theorem}

If the Killing tensor is exact rather than conformal, then we can use Eq. (46b) of Ref. \cite{Kobialko:2022ozq} resulting in $\lambda_\tau = 0$, which implies that $\mathcal{H}^{ab}$, $\mathcal{J}^a$ and $\mathcal{M}^2$ are constant along the hypersurface, or equivalently 
\begin{align}
  q\mathcal{H}^{ab} \LCS_\tau  p_{b} = q\LCS_\tau\left(
          \frac{1}{2\chi_\tau} \mathcal{F}^{a} \right), \quad
q^2 \mathcal{F}^{a} \LCS_\tau p_a =0.
\end{align}
In particular, for $q=0$ the slices are always glued particle surfaces. This generalizes theorems from Refs. \cite{Kobialko:2022ozq,Kobialko:2021aqg} that slices associated with a slice-reducible exact Killing tensor are fundamental photon surfaces to the case of massive particle surfaces with neutral particles. 

Now, we can analyze the situation in an opposite direction, i.e., whether the existence of foliation with slices representing glued massive particle surfaces results in the existence of an exact Killing tensor of rank two. If there is a foliation of glued massive particle surface of maximum order with $\beta^a{}_\alpha = 0$, Eq. (\ref{eq:Dcondition_a}) corresponds to one of the necessary conditions for the existence of slice-reducible Killing tensor from Ref. \cite{Kobialko:2022ozq} if the lapse function $\varphi$ of the foliation satisfies the condition $\lambda_\gamma = - \LCS_\gamma\ln\left(\chi_\tau\varphi^3\right)$ (compare with conditions (34) in Ref. \cite{Kobialko:2022ozq}). If it additionally meets another condition 
\begin{align} 
\LCS_\gamma \left( \varphi \chi_\tau - \varphi n^\alpha \LCM_{\alpha} \ln \varphi \right)
    = 0,
\end{align}
we immediately get the evidence of the existence of the slice-reducible conformal Killing tensor of rank two.

If the spacetime possesses slice-reducible conformal (or exact) Killing tensor, following Eqs. (24a) and (29) from Ref. \cite{Kobialko:2022ozq}, there are two differential equations
\begin{equation}
    n^\alpha \nabla_\alpha \Psi = 2\chi_\tau,\qquad
    n^\alpha \nabla_\alpha \left(e^\Psi \G^{ab}\right) = n^\alpha \nabla_\alpha \gamma^{ab}.
\end{equation}
These equations are useful to express $\mathcal{H}^{ab}$ in form (\ref{eq:HasG}) through the components of the Killing tensor
\begin{align}\label{eq:HasGamma}
    \mathcal{H}^{ab} =
    - \frac{e^{-\Psi}}{2\chi_\tau} n^\alpha  \nabla_\alpha \gamma^{ab}.
\end{align}

The presence of glued particle surfaces can serve as an indication of the existence of an exact Killing tensor. We will illustrate the potential of this representation for metrics such as Kerr and others in the following examples.

\section{Examples}
\label{sec:Examples}

\subsection{Kerr metric} 

The metric for Kerr solution in the Boyer-Lindquist coordinates reads
\begin{equation}\label{eq:kerr_form}
    ds^2 =
    - f(dt - \omega d\phi)^2
    + \frac{\Sigma}{\Delta}dr^2
    + \Sigma d\theta^2
    + \Delta f^{-1}\sin^2\theta d\phi^2,
\end{equation}
where
\begin{equation}
    f = \frac{\Delta - a^2\sin^2\theta}{\Sigma}, \qquad
    \Sigma = r^2 + a^2 \cos^2\theta,\qquad
    \omega=\frac{-2Mar\sin^2\theta}{\Delta-a^2\sin^2\theta},\qquad
    \Delta = r(r-2M) + a^2.
\end{equation}
The Kerr spacetime possesses two Killing vectors $k_{(t)}^\alpha\partial_\alpha=\partial_t$, $k_{(\phi)}^\alpha\partial_\alpha=\partial_\phi$, and an exact Killing tensor in the form (\ref{eq:sr_killing}) with
\begin{equation}\label{eq:kerr_killing_r}
    \alpha = -r^2,\qquad
    \gamma^{ab}\kappa_a^\alpha\kappa_b^\beta = -\Delta^{-1} S^\alpha S^\beta,\qquad
    S^\alpha = \left(r^2 + a^2\right) \delta^\alpha_t + a \delta^\alpha_\phi,\qquad
    e^\Psi = \Sigma,
\end{equation}
which is slice reducible with respect to slices $r=\text{const}$ with the normal unit vector ${n^\alpha = \sqrt{\Delta/\Sigma} \delta^\alpha_r}$. The slices are tangent to the Killing vectors $\partial_t$, $\partial_\phi$. According to the result obtained in Ref. \cite{Kobialko:2022ozq}, Sec. V.A, four-dimensional solutions with two commuting Killing vectors $n_K = n - 2$ have two different foliations orthogonal to each other that generate the same Killing tensor. The tensor (\ref{eq:kerr_killing_r}) is indeed slice-reducible with respect to $\theta=\text{const}$ as well, possessing another representation of the form (\ref{eq:sr_killing})
\begin{equation}\label{eq:kerr_killing_r_2}
    \alpha = -a^2\cos^2\theta,\qquad
    \gamma^{ab}\kappa_a^\alpha\kappa_b^\beta = S^\alpha S^\beta,\qquad
    S^\alpha = a\sin\theta \delta^\alpha_t + \delta^\alpha_\phi / \sin\theta,\qquad
    e^\Psi = \Sigma
\end{equation}
with ${n^\alpha = \sqrt{1/\Sigma} \delta^\alpha_\theta}$. As we have shown, slices associated with a slice-reducible Killing tensor are glued massive particle surfaces for neutral particles. Since the Maxwell form is trivial $A_\alpha = 0$ in the Kerr solution, the electric charge $q$ plays no role in the particles' motion. It motivates us to consider these two cases. In the first case, $r=\text{const}$, spacetime sections of the hypersurfaces are compact and the corresponding particles are bounded. In the second case $\theta=\text{const}$, spacetime sections are not compact, and as we will see, particles fly by the black hole in the scattering regime.

\subsubsection{Spheres \texorpdfstring{$r=\text{const}$}{r-const}}

From the $\theta\theta$-component we get $\chi_\tau$, and comparing the rest of $\chi_{\alpha\beta}$ with $h_{\alpha\beta}$, we get the representation (\ref{eq:condition_3}) of the second fundamental form with\footnote{We use the following index ordering: $t,\,\varphi$. }
\begin{subequations}
\begin{equation}
    \chi_\tau = \chi_{\theta\theta} / h_{\theta\theta} = 
    r \sqrt{\frac{\Delta}{\Sigma^3}},
\end{equation}
\begin{equation}
    \mathcal{H}^{ab} =
    \frac{1}{r \Delta^2}
    \begin{pmatrix}
        \left(a^2+r^2\right) \left(a^2 (M+r)+r^2 (r-3 M)\right)
        &
        a M \left(a^2-r^2\right)
        \\
        a M \left(a^2-r^2\right)
        &
        a^2 (M-r)
    \end{pmatrix}.
\end{equation}
\end{subequations}
One can check that Eq. (\ref{eq:HasGamma}) does hold. Matrix $\mathcal{H}^{ab}$ is explicitly independent on $\theta$, thus it provides the maximal gluing. Since the determinant of $\mathcal{H}^{ab}$ is negative
\begin{equation}
    \text{det} \mathcal{H}^{ab} = -\left(\frac{a}{r (r-2 M) + a^2}\right)^2,
\end{equation}
constraint (\ref{eq:condition_q}) describes a hyperbola. Constraint (\ref{eq:condition_q}) and the constraint on the vector norm $\dot{\gamma}^2 = -m^2$ explicitly read
\begin{subequations}
\begin{align} \label{eq:kerr_r}
    &
    \left(E (a^2 + r^2) - a L_z \right)
    \left(E a^2 (M+r) + E r^2 (r-3 M) + a L_z (r-M) \right)
    - r \Delta^2 m^2 = 0,
    \\
    &
      (r^2 + a^2) \left(\mathcal{Q}-(L_z-a E)^2 - r^2 E^2\right)
    - 2 M \mathcal{Q} r 
    + r^2 (L_z^2 + \Delta m^2)
    = 0,
\end{align}
\end{subequations}
where 
\begin{equation}
    E = -\dot{\gamma}^\alpha k_{(t)}{}_\alpha,\qquad
    L_z = \dot{\gamma}^\alpha k_{(\phi)}{}_\alpha,\qquad
    \mathcal{Q} = -\dot{\gamma}^\alpha \dot{\gamma}^\beta \mathcal{K}_{\alpha\beta},\qquad
    q_a = \{-E,L_z\}.
\end{equation}
There are two families of solutions $(E^+, L_z^+)$ and $(E^-, L_z^-)$ to Eqs. (\ref{eq:kerr_r}), where
\begin{subequations}
\begin{align}
    &
    E^\pm =
    \pm \frac{
        m^2 \left(r \left(a^2+2 r^2\right)-3 M r^2\right) + \mathcal{Q}(r - M)
    }{
        2 r \sqrt{\Delta \left(m^2 r^2+\mathcal{Q}\right)}
    },
    \\&
    L_z^\pm =
    \pm \frac{
          m^2 r \left(a^2 r (r-3 M)+a^4+M r^3\right)
        - \mathcal{Q} \left( r^2 (r - 3 M) + a^2 (r + M)\right) 
    }{
        2 a r \sqrt{\Delta \left(m^2 r^2+\mathcal{Q}\right)}
    }.
\end{align}
\end{subequations}
Here, the Carter constant plays the role of the parameter of the one-dimensional family of the solutions. Thus, the order of the glued particle surface is one. The region accessible for particles is determined by inequality (\ref{eq:accessible_surface}), which takes the following form for Kerr solution
\begin{equation} \label{eq:kerr_condition}
      \frac{f(L_z-E \omega)^2}{\Delta \sin^2\theta }
    - \frac{E^2}{f}
    \leq -m^2.
\end{equation}
The left-hand side diverges to $+\infty$ at the poles of the sphere $\theta = 0$ or $\pi$ if $L_z\neq0$. In this case, the massive particle surfaces represent belts cut out from the sphere. 

The resulting extended surfaces of maximal order for massive particles are constructed to include sets of spherical massive particle worldlines. In the case of chargeless particles, these surfaces contain spherical timelike geodesics discussed in Teo's paper \cite{Teo:2020sey}. Therefore, we have obtained the geometric definition of Teo's spheres. Our expressions coincide with Teo's results after shifting the Carter's constant $\mathcal{Q} = \mathcal{Q}_{Teo} + (L_z - a E)^2$. Teo has shown the existence of both stable and unstable orbits on the same sphere depending on the value of $Q$. The marginally stable orbit is determined by Carter's constant
\begin{align}
\mathcal{Q} = \mathcal{Q}_{MS} \equiv m^2 \frac{r^3 (a^2-M^2) \pm \sqrt{M\Delta^3  r^3}}{r (r-M)^2-\Delta  M}.
\end{align}
Thus, glued massive particle surfaces can contain particles with both stable and unstable worldlines. However, the stability of each individual surface in the family is well-defined. One can approach the problem of worldline stability on massive particle surfaces in a general manner by utilizing calculations similar to those described in Ref. \cite{Kobialko:2022uzj}. 

A different point of view on glued massive particle surfaces is also useful. For each value of the momentum $L_z$ at fixed energy $E$, there exists a unique massive particle surface, and the entire set of surfaces for all $L_z$ forms a region similar to the photon region  \cite{Grenzebach:2014fha,Grenzebach:2015oea,Galtsov:2019bty,Galtsov:2019fzq}. The structure of this region is the key for describing the shadow of a given black hole in a stream of massive charged particles on a given energy scale in future works.

\subsubsection{Cones \texorpdfstring{$\theta = \text{const}$}{theta-const}}

Similarly, for cones with constant $\theta$ we get 
\begin{equation}
    \chi_\tau = \chi_{rr} / h_{rr} = 
    -\frac{a^2 \sin (2\theta )}{2\Sigma^{3/2}},\qquad
    \mathcal{H}^{ab} =
    \begin{pmatrix}
        1 & 0 \\
        0 & -1/a^2\sin^4\theta
    \end{pmatrix}.
\end{equation}
One can check that Eq. (\ref{eq:HasGamma}) does hold. Conditions on the integrals of motion
\begin{subequations}
\begin{align}
    &
    \frac{L_z^2}{a^2\sin^4\theta}-E^2+m^2 = 0,
    \\ & 
      \frac{ \mathcal{Q}' }{a^2 \cos^2\theta}
    + \frac{ L_z^2 }{a^2 \sin^2\theta}
    - E^2
    + m^2 = 0
\end{align}
\end{subequations}
have the following solution
\begin{align}
    L_z^2 = \mathcal{Q}' \tan^4\theta,\qquad
    E^2 = m^2 + \frac{\mathcal{Q}'}{a^2\cos^4\theta}.
\end{align}
where $\mathcal{Q}' = \mathcal{Q} + (L_z-a E)^2$ is a modified Carter constant. As a result, we find that $E^2\geq m^2$. The condition for the accessible region has the same form (\ref{eq:kerr_condition}).
It is noteworthy that $\mathcal{H}^{ab}$ is nonexistent for $a=0$ due to the fact that, as $a$ approaches 0, all worldlines on the cone degenerate into radial geodesics without any gluing. By construction, all resulting $\theta = \text{const}$ massive particles cones contain conical orbits which are generally non-compact, corresponding to the minimum of the effective potential in the equations for $\theta$ from the geodesic equations \cite{Teo:2020sey}. One can expect the existence of compact conical orbits in the charged case in Kerr-Newman dyons (see Ref. \cite{Hackmann:2013pva} for an example of compact orbits between two closely spaced cones). While these non-compact orbits may not hold the same level of significance as compact spherical orbits in the formation of gravitational shadows or in the context of Penrose inequalities, they offer possibilities for investigating particle scattering problems.

From the expression of the second derivative of the effective potential for $u=\cos\theta$ (see Ref. \cite{Teo:2020sey})
\begin{align}
V''(u) = - 8 a^2 (E^2- m^2) u^2\leq0
\end{align}
follows that all massive particle surfaces $\theta = \text{const} \neq \pi/2$ are stable. In contrast to spherical surfaces, this property remains unaffected by Carter's constant $\mathcal{Q}$.

\subsection{Zipoy-Voorhees} 
Zipoy-Voorhees solution is a vacuum non-spherical spacetime with the following metric 
\begin{align}
    ds^2 = -f^\delta dt^2 + k^2 f^{-\delta} \Bigg(
        &\left(\frac{x^2 - 1}{x^2 - y^2}\right)^{\delta^2}(x^2-y^2)\left(\frac{dx^2}{x^2-1} + \frac{dy^2}{1-y^2}\right)
        +\\\nonumber &
        + (x^2-1)(y^2-1) d\phi^2
    \Bigg),\qquad
    f = \frac{x-1}{x+1},
\end{align}
with $x > 1$ and $y^2 \leq 1$. Zipoy-Voorhees is known to lack photon surfaces, but it has fundamental photon surfaces. Thus, we expect that there are massive particle surfaces, which are not glued. We take surfaces $x = \text{const}$ as an ansatz to investigate whether this spacetime allows for spherical massive particle surfaces. Following a similar approach to previous examples, we obtain
\begin{equation}\label{eq:zv_H}
    \mathcal{H}^{ab} = 
    \frac{1}{
        F_1
    }
    \left(
        \begin{array}{cc}
         f^{-\delta } (F_1 - \delta(x^2-y^2)) & 0 \\
         0 & \cfrac{\left(1-\delta ^2\right) x f^{\delta }}{k^2 \left(x^2-1\right)} \\
        \end{array}
    \right),
\end{equation}
where
\begin{equation}
    F_1 = x^3-\delta  x^2+x \left(\delta ^2 \left(1-y^2\right)-1\right)+\delta  y^2.
\end{equation}
Substituting (\ref{eq:zv_H}) and integrals of motion $Q_a = q_a = \{-E, L_z\}$ into Eq. (\ref{eq:condition_q}), we find that the expression (which is quite voluminous) does not depend on $y$ if
\begin{equation}
    L_z^2 = \frac{E^2 k^2 \left(x^2-1\right)^2 f^{-2\delta }}{1-\delta  x}.
\end{equation}
The final solution is
\begin{align}
    E^2 = m^2 f^{\delta} \frac{1-\delta  x}{2-\delta  x},\qquad
    L_z^2 = k^2 m^2 f^{-\delta} \frac{\left(x^2-1\right)^2 }{2-\delta  x}.
\end{align}
Condition (\ref{eq:accessible_surface}) reads
\begin{equation}
    \frac{m^2 (x^2-y^2)}{\left(1-y^2\right) (2-\delta  x)} \leq 0.
\end{equation}
Along with this condition, we require non-negativeness of $E^2$ and $L_z^2$. However, this is not possible for timelike geodesics, but it is possible for tachyons $m^2 < 0$ only.

\subsection{Kerr-Newman dyon}
Dyonic Kerr-Newman solution is an electrovacuum solution with electric $\chQ$ and magnetic $\chP$ charges. Its metric has the form (\ref{eq:kerr_form}) with redefined functions
\begin{equation}
    \Delta = r(r-2M) + a^2 + \chSq,\qquad
    \omega=\frac{-2a\sin^2\theta}{\Delta-a^2\sin^2\theta}\left(Mr - \chSq/2\right)
\end{equation}
and electromagnetic vector potential
\begin{equation}
    A_\alpha dx^\alpha = 
      \frac{\chQ r + a \chP \cos\theta}{\Sigma} dt
    - \frac{\chP \left(a^2+r^2\right) \cos\theta + a \chQ r \sin^2\theta}{\Sigma} d\phi,
\end{equation}
where $\chSq = \chQ^2 + \chP^2$. We can easily generalize our formalism to particles both with  electric $q$ and magnetic $p$ charges as follows. First, we notice if the particle possesses a magnetic charge, the equation of motion has the form of Eq. (\ref{eq_particles}) with the following substitution
\begin{equation} 
   F^\beta{}_{\lambda} \to \tilde{F}^\beta{}_{\lambda} =  F^\beta{}_{\lambda} + \frac{p}{q} (\star F)^\beta{}_{\lambda},\qquad
   (\star F)_{\mu\nu} = \frac{1}{2} \epsilon_{\mu\nu\alpha\beta} F^{\alpha\beta},
\end{equation} 
where $\epsilon_{\mu\nu\alpha\beta}$ is the Levi-Civita tensor. Since tensor $F_{\alpha\beta}$ fulfills both Maxwell and Bianchi equations, its dual tensor $(\star F)_{\alpha\beta}$ can be represented with a vector potential $(\star F)_{\alpha\beta} = \partial_\alpha A^\star_\beta - \partial_\beta A^\star_\alpha$ as well. Thus, we can make the following substitution of the vector-potential in our analysis
\begin{equation}
      A_\alpha \to \tilde{A}_{\alpha} =  A_\alpha + \frac{p}{q} A_\alpha^\star.
\end{equation}
A detailed analysis of the motion of dyonic particles in the dyonic Kerr-Newman solution is presented in Ref. \cite{Hackmann:2013pva}. Due to the dual nature of dyons, the transformed potential $\tilde{A}_{\alpha}$ is the same as $A_{\alpha}$ up to substitution
\begin{equation}
    \chQ \to \chQ + \chP p/q,\qquad
    \chP \to \chP - \chQ p/q.
\end{equation}
Performing the same steps as for previous examples, and making use of Eqs. (\ref{eq:jm}), we have the following quantities for spheres $r=\text{const}$
\begin{subequations}
\begin{equation}
    \mathcal{H}^{ab} =
    \frac{1}{r \Delta^2}
    \begin{pmatrix}
        \left(a^2+r^2\right) \left[2 r \Delta + \left(a^2+r^2\right) (M-r)\right]
        &
        a\left[a^2 M - r \left(M r-\chSq\right)\right]
        \\
        a\left[a^2 M - r \left(M r-\chSq\right)\right]
        &
        a^2 (M-r)
    \end{pmatrix},
\end{equation}
\begin{equation}
    \mathcal{J}^{a} =
    - \frac{(\chP p + \chQ q)}{2r \Delta^2}
    \begin{pmatrix}
        r^4 - 4 M r^3 + (2 a^2 + 3 \chSq) r^2 + a^2 (a^2 + \chSq)
        \\
        a \left(a^2+\chSq-r^2\right)
    \end{pmatrix},
\end{equation}
\begin{equation}
    \mathcal{M}^2 =
    m^2 + \frac{(\chP p+\chQ q)^2(r^2 - M r - \Delta)}{\Delta^2},
\end{equation}
\end{subequations}
and for cones $\theta=\text{const}$
\begin{subequations}
\begin{equation}
    \mathcal{H}^{ab} =
    \begin{pmatrix}
        1
        &
        0
        \\
        0
        &
        -1/a^2\sin^4\theta
    \end{pmatrix},
\end{equation}
\begin{equation}
    \mathcal{J}^{a} =
    \frac{\chQ p - \chP q}{2 a^2 \cos\theta \sin^4\theta}
    \begin{pmatrix}
        a \sin^4\theta
        \\
        1 + \cos^2\theta
    \end{pmatrix},
\end{equation}
\begin{equation}
    \mathcal{M}^2 =
    m^2 + \left(
        \frac{\chQ p - \chP q}{a\sin^2\theta}
    \right)^2.
\end{equation}
\end{subequations}
All quantities are constant along the surface, so spheres and cones achieve maximal gluing. It is notable that particle charges come only in the combination $\chP p+\chQ q$ for spheres, and $\chQ p - \chP q$ for cones. Making the first or the second combination equal to zero (i.e., $p = -q\chQ/\chP$ or $p = q\chP/\chQ$), particles on the corresponding surfaces will not be affected by the electromagnetic field, since the electric and magnetic forces balance each other out. Particularly, if $p=\chP=0$, conical surfaces are not affected, because the electromagnetic field exerts only the radial electric force, which has no influence on the motion normal to the cones. Similarly, if $p = \chQ = 0$, spheres are not affected, because the electric force is absent, and the magnetic force is tangent to the surface.

\section{Discussion and Conclusions}

We have generalized the concept of photon surfaces and fundamental photon surfaces to the case of massive charged particles or photons in the medium \cite{Perlick:2017fio,Asenjo:2010zz} in general spacetimes with at least one isometry, among which are stationary ones. We introduced the notion of massive particle surfaces as timelike hypersurfaces such that any worldline of a particle with mass $m$, electric charge $q$, and a set of conserved quantities $Q_a$, remains tangent to the surface forever, provided that the worldline initially was tangent to it. This concept encompasses the generalizations for both ordinary photon spheres and all their standard generalizations \cite{Yoshino:2017gqv,Yoshino:2019dty}.  In general, such surfaces have boundaries, but the notion can be naturally extended to surfaces without boundaries, which can be useful for analyzing their topological properties and integrability problems.
In particular, for the analysis of areal constraints, Riemannian Penrose inequalities, one can try to apply the standard approach of reverse flows of mean curvature \cite{Yoshino:2017gqv}.

The new concepts introduced here represent a natural extension of our prior research, as presented in Refs. \cite{Kobialko:2022uzj,Kobialko:2020vqf}. In particular, we have proved the key theorem \ref{th:theorem}, which is a direct analog of theorem 2.2 obtained in Ref. \cite{Claudel:2000yi} for photon surfaces, and establishes a purely geometric definition of such surfaces. We have found that in the general case the second fundamental form $\SFS_{\alpha\beta}$ is partially umbilical. The condition of the partially umbilical surface becomes degenerate when at least $n-2$ Killing vectors are involved, resulting in the fixation of the function $\chi_\tau$ without any additional conditions. Also, we have obtained restrictions on the electromagnetic field tensor and a compact master equation.

We have also found that in spacetimes with hidden symmetries (Kerr-like spacetimes) described by slice-reducible conformal Killing fields of rank two, massive particle surfaces can be found among the slices of the corresponding foliation. Moreover, such surfaces are  glued together in a non-trivial way. We have defined the general concept of glued massive particle surfaces of arbitrary order. Glued surfaces capture worldlines from some continuous family of conserved quantities of dimension larger than or equal to one. We have reinforced these concepts with a number of examples demonstrating that the previously known features of the particle movement along worldlines are indeed successfully explained by the geometric properties of the glued surfaces. We have investigated the relationship between the gluing phenomenon and the formation of a slice-reducible conformal Killing tensor. Particularly, we have found that the existence of a foliation of glued massive particle surfaces of maximum order can serve as a good indication of the existence of the Killing tensor, and vice versa. However, a rigorous substantiation of this correspondence remains a part of further research. Despite this, we can already retrieve one of the necessary and sufficient integrability conditions for the existence of conformal Killing tensors.

If an exact Killing tensor of rank two is present  in spacetime, there is a quantity $\mathcal{H}^{ab}$ extracted from the second fundamental form, which is constant along the surface. If massive particle surfaces allow for charged particles, there are two more constant quantities, $\mathcal{J}^a$, and $\mathcal{M}^2$. Together they can be understood as invariants of the massive particle surface.

A number of examples are called to demonstrate the usefulness of the concept. We have considered the problem of the existence of massive particle surfaces in Kerr, Zipoy-Voorhees and Kerr-Newman metrics. The obtained results can be generalized to particles coupled to background fields in various theories by replacing of $m^2$ and $A_\alpha$ with the corresponding expressions. The obtained results have potential applications in various areas, including Penrose inequalities, uniqueness theorems, integrability theory, and the description of massive particle dynamics in the presence of real supermassive objects, as well as for further photons in plasma.

\begin{acknowledgments}
This work was supported by Russian Science Foundation under Contract No. 23-22-00424.
\end{acknowledgments}

\appendix
\section{Quadratic forms with linear constraints}\label{sec:appendix_theorems}

Here we prove several useful theorems about quadratic forms with linear constraints.

\begin{theorem}\label{th:pre_theorem}
Let $\mathbf{Y} = \{\y\}$ be a non-empty set of all solutions for the following quadratic equation
\begin{equation}\label{eq:pre_theorem_system_1}
    \y^T Q \y = d,
\end{equation}
and similarly $\mathbf{Y}' = \{\y'\}$ is a set of all solutions for
\begin{equation}\label{eq:pre_theorem_system_2}
    \y^T Q' \y = d',
\end{equation}
where  $Q, Q'$ are symmetric non-degenerate $n_Q\times n_Q$-matrices. Then the set of all solutions for the first quadratic equation (\ref{eq:pre_theorem_system_1}) is contained in the set of solutions for the second quadratic equation (\ref{eq:pre_theorem_system_2}), $\mathbf{Y} \subseteq \mathbf{Y}'$, if and only if the primed quantities for the second system can be represented as 
\begin{align}  \label{eq:pre_theorem_decomposition}
    Q' &= \lambda Q,\qquad
    d' = \lambda d,
\end{align}
where $\lambda$ is an arbitrary scalar\footnote{Including  $\lambda=0$. In case $\lambda\neq0$ we find that $\mathbf{Y} =\mathbf{Y}'$.}.
\end{theorem}

\begin{proof}

\textit{Sufficient $(\Leftarrow)$}: Let $\y$ be an arbitrary solution of the first quadratic equation (\ref{eq:pre_theorem_system_1}). Then, conditions (\ref{eq:pre_theorem_decomposition}) give
\begin{equation}
    \y^T Q' \y = \lambda \y^T Q \y=\lambda d =d',
\end{equation}
i.e., exactly the second quadratic equation (\ref{eq:pre_theorem_system_2}). 

\textit{Necessary $(\Rightarrow)$}: Let all solutions $\y$ for the first quadratic equation (\ref{eq:pre_theorem_system_1}) be contained in the set of solutions for the second quadratic equation (\ref{eq:pre_theorem_system_2}). Let us introduce an orthonormal basis for $Q$ in the following way
\begin{equation}
   \e^T_i Q \e_j=-\delta_{ij}, \quad \e^T_{i'} Q \e_{j'}=\delta_{i'j'}, \quad \e^T_{i} Q \e_{i'}=0. 
\end{equation}
Consider a subset of the solutions of the system (\ref{eq:pre_theorem_system_1}) defined as follows:
\begin{equation}
 \y = a (\e_i \cos \psi +\e_j \sin \psi)  + b (\e_{i'}\cos \psi'+\e_{j'} \sin \psi'), \quad -a^2+b^2 =d.
\end{equation}
Writing quadratic form $Q'$ in this orthonormal basis $Q'_{ij}\equiv \e^T_i Q' \e_j$, Eq. (\ref{eq:pre_theorem_system_2}) for chosen $\y$ gives
\begin{align}
    &a^2 \left(\frac{Q'_{ii}+Q'_{jj}}{2}+ Q'_{ij}\sin (2 \psi) +\frac{Q'_{ii}-Q'_{jj}}{2} \cos (2 \psi)\right)\nonumber\\&+b^2 \left(\frac{Q'_{i'i'}+Q'_{j'j'}}{2}+ Q'_{i'j'}\sin (2 \psi') +\frac{Q'_{i'i'}-Q'_{j'j'}}{2} \cos (2 \psi')\right)\nonumber\\&+2 ab \left(Q'_{ii'}\cos \psi \cos \psi'+Q'_{ij'}\cos \psi \sin \psi'+Q'_{ji'}\sin \psi  \cos \psi' +Q'_{jj'}\sin \psi \sin \psi'\right) = d' \nonumber.
\end{align}
So from the arbitrariness of $\psi$ and $\psi'$
\begin{align}
a^2 Q'_{ii}+b^2 Q'_{i'i'}=d', \quad
Q'_{ii}=Q'_{jj}, \quad  Q'_{i'i'}=Q'_{j'j'}, \quad Q'_{ij}=Q'_{i'j'}=Q'_{ii'}=0.
\end{align}
Eliminating $a$ from the first equation, we get
\begin{align}
b^2 (Q'_{i'i'}+Q'_{ii})=d' + d Q'_{ii}, 
\end{align}
From the arbitrariness of $b$ follows
\begin{align}
-Q'_{ii}=Q'_{i'i'}= d'/d\equiv \lambda,
\end{align}
in orthonormal basis for $Q$. So $Q'$ must be proportional to $Q$ with the same proportionality factor as $d'$ and $d$:
\begin{equation}
    Q' = \lambda Q, \quad d'=\lambda d.
\end{equation}
\end{proof}

\begin{theorem}\label{th:solutions}
Let $\mathbf{X} = \{\x\}$ be a non-empty set of all solutions for the following quadratic equation with a system of linear constraints
\begin{equation}\label{eq:theorem_system_1}
    \x^T A \x = d,\qquad B \x = \f,
\end{equation}
and similarly $\mathbf{X}' = \{\x'\}$ is a set of all solutions for
\begin{equation}\label{eq:theorem_system_2}
    \x^T A' \x + 2 \w'^T \x = d',\qquad B' \x = \f',
\end{equation}
where  $A, A'$ are symmetric $n_A\times n_A$-matrices, $B, B'$ are $n_B \times n_A$-matrices and $n_B$-vector $\f$ is in the image of $B$, and $A$ is non-degenerate. If all solutions $\mathbf{X}$ projected onto $\text{ker}B$ with respect to the pseudo-inner product $A$ span the entire $\text{ker}B$, then the following statement is fair. The entire set of solutions for the first system (\ref{eq:theorem_system_1}) is contained in the set of solutions for the second system (\ref{eq:theorem_system_2}), $\mathbf{X} \subseteq \mathbf{X}'$, if and only if the primed quantities for the second system can be represented as 
\begin{subequations} \label{eq:theorem_decomposition}
\begin{align}
    A' &= \lambda A+ S_1 B + (S_1 B)^T+ B^T S_2 B,\\
        B'& = S_3 B, \\  
      \f' &= S_3 \f,
\end{align}
\end{subequations}
where $\lambda$ is an arbitrary scalar, matrices $S_1, S_2, S_3$ are arbitrary matrices constrained by the following conditions
\begin{align}\label{eq:theorem_constrain}
     \f^T S_2 \f = d' - \lambda d  - 2 \f^T G B A^{-1} \w',\quad
      (\mathds{1} - B^T G B A^{-1}) \w' + S_1 \f = 0, \quad
      B A^{-1} S_1 = 0,
\end{align}
where $G$ is a pseudo-inverse matrix of $G$ defined from the following operator equation $(B A^{-1} B^T) G B = B$.
\end{theorem}

\begin{proof}

\textit{Sufficient $(\Leftarrow)$}: For an arbitrary solution $\x$ of system (\ref{eq:theorem_system_1}), conditions (\ref{eq:theorem_decomposition}) and (\ref{eq:theorem_constrain}) from the theorem give 
\begin{align}
  \x^T A' \x + 2\w'^T \x &=\lambda \x^T A\x+ \x^T S_1 B\x + \x^T (S_1 B)^T\x+ \x^T B^T S_2 B \x + 2\w'^T \x=\\\nonumber
  &=\lambda d + 2\x^T S_1 \f + \f^T S_2 \f + 2\w'^T \x =\\\nonumber & =
   d' +  2\x^T S_1 \f  - 2 \f^T G B A^{-1} \w' - 2\f^T S_1^T \x + 2\w'^T A^{-1} B^T G B \x
  =\\\nonumber & =
   d'  - 2 \f^T G B A^{-1} \w' + 2\w'^T A^{-1} B^T G \f
  =d', 
\end{align}
and 
\begin{equation}
 B' \x =S_3 B \x =S_3 \f= \f'. 
\end{equation}
Thus, every solution of the first system (\ref{eq:theorem_system_1}) is a solution for the second system (\ref{eq:theorem_system_2}).

\textit{Necessary $(\Rightarrow)$}: For simplicity we define the spaces $U_\x$ and $U_\f$, so
\begin{align}
&
    \w', \x \in U_\x,\qquad \f \in U_\f,
\\\nonumber &
    A, A': U_\x \mapsto U_\x,\qquad
    B, B': U_\x \mapsto U_\f.
\end{align}
We can understand matrix $A$ as a metric tensor in $U_\x$. It defines the pseudo-inner products $\x_1^T A \x_2$. We will understand orthogonality as $\x_1^T A \x_2 = 0$. Further, the subspace defined by the kernel of $B$ will play an important role. We will define the kernel of matrix $B$ and the orthogonal complement of the kernel $(\text{ker}B)^\perp$ as
\begin{equation}
    \text{ker}B=\{\y: B\y = 0\},\qquad 
    (\text{ker}B)^\perp = \{\x: \x^T A \y = 0, \forall\y\in\text{ker}B\}.
\end{equation}
We define a pseudo-inverse matrix $G$ from the following operator equation 
\begin{equation}\label{eq:G_def}
    (B A^{-1} B^T) G B = B.
\end{equation}
This matrix is useful to solve linear equations with matrix $B$. For example, since $\f$ is in the image of $B$, we have the following identity $(B A^{-1} B^T) G \f = \f$. One can define $G$ as an inverse matrix $(B A^{-1} B^T)^{-1}$ if the latter exists. Also, we can introduce a projector at the subspace $\text{ker}B$ in the following way
\begin{equation}\label{eq:projector_pi}
    \Pi = \mathds{1} - A^{-1} B^T G B.
\end{equation}
Indeed, it is a projector, since it preserves the vector from the kernel $\Pi \y = \y$, and any vector $\x \in (\text{ker}B)^\perp$ orthogonal to the kernel is zero 
\begin{equation}
    \Pi (A^{-1} B^T \bm{t}) = A^{-1} B^T \bm{t} - A^{-1} (BA^{-1} B^T G B)^T \bm{t} = A^{-1} B^T \bm{t} - A^{-1} B^T \bm{t} = 0,
\end{equation}
where $(A^{-1} B^T \bm{t})$ can represent any vector from $(\text{ker}B)^\perp$. To show this, we notice that $A^{-1}$ is a full-rank matrix, therefore the image of the operator $A^{-1} B^T$ has the same dimension as the orthogonal complement of the kernel of $B$
\begin{equation}
    \text{rank}(A^{-1} B^T) = \text{rank}B = \text{dim}(\text{ker}B)^\perp,
\end{equation}
and the result is orthogonal to any vector $\y\in\text{ker}B$
\begin{equation}
    \y^T A A^{-1} B^T \bm{t} = (B\y)^T \bm{t} = 0.
\end{equation}
Thus, the image of $(A^{-1} B^T)$ is $(\text{ker}B)^\perp$.

First, vector $\x$ and matrices $A'$, $B'$ can be split into terms that are spanned onto the kernel $\text{ker}B \subseteq U_\x$ and its orthogonal complement $\left(\text{ker}B\right)^\perp\subseteq U_\x$:
\begin{align}\label{eq:AB_representation}
    \x = \z + \y,\qquad
    A' = S_0 + S_1 B + (S_1 B)^T + B^T S_B B,\qquad
    B' = S_4 + S_3 B,
\end{align}
and we have imposed the following conditions
\begin{align}\label{eq:S_conditions}
    B \y=0,\qquad
    \z^T A \y = 0,\qquad
    S_0 A^{-1} B^T = B A^{-1} S_0 = 0, \quad B A^{-1} S_1 = 0,\qquad
    S_4 A^{-1} B^T = 0.
\end{align}
Matrices $S_0$ and $S_B$ are symmetric. Constraints (\ref{eq:S_conditions}) are orthogonality conditions with respect to the aforementioned pseudo-inner product. The new linear operators act as follows:
\begin{align}
&
    S_0: U_\x \mapsto U_\x,\qquad
    S_1: U_\x \mapsto U_\f,\qquad
    S_4: U_\x \mapsto U_\f,\qquad
    S_B, S_3: U_\f \mapsto U_\f.
\end{align}

Substituting the definition of $\x$ into the linear system of Eqs. (\ref{eq:theorem_system_1}), we get 
\begin{subequations} \label{eq:solution_A}
\begin{equation}\label{eq:solution_A_1}
   B \x = B\z = \f
   \qquad\Rightarrow\qquad
   \z=A^{-1} B^T G\f.
\end{equation}
The term $A^{-1} B^T G \f$ can be understood as a particular solution to the linear system $B\x=\f$, and $\y$ is a solution to the homogeneous part of the linear system $B\x=0$.  Substituting the definition of $\x$ and $\z$ into the quadratic equation in (\ref{eq:theorem_system_1}), we get
\begin{equation}\label{eq:solution_A_2}
    \x^T A \x = \f^T G \f + \y^T A \y = d
    \qquad\Rightarrow\qquad
    \y^T A \y = d-\f^T G \f,
\end{equation}
\end{subequations}
The vector $\y$ is in $\text{ker} B$, thus we can use the projected operator $A_\text{pr}$ in the quadratic condition
\begin{equation}\label{eq:theorem_condition_5_2}
    \y^T A_\text{pr} \y = d - \f^T G \f,
\end{equation}
where $A_\text{pr} = \Pi^T A \Pi = A - B^T G B$ is a projected operator $A$ onto $\text{ker}B$. The operator $A_\text{pr}$ has the following properties: 
\begin{itemize}
    \item $A_\text{pr} \y = A \y$ for all $\y\in\text{ker}B$. Hence, $A_\text{pr}$ acts on $\y$ the same way as $A$, and the quadratic form $\y^T A_\text{pr} \y$ is identical to $\y^T A \y$ in $\text{ker}B$.
    \item $A_\text{pr} A^{-1}B^T = B A^{-1} A_\text{pr} = 0$, so $A_\text{pr}$ satisfies all conditions on $S_0$ in Eq. (\ref{eq:S_conditions}).
\end{itemize}
The Eq. (\ref{eq:theorem_condition_5_2}) may have no solutions at all, which is not the case for the theorem.

Substituting $\x$ into the system (\ref{eq:theorem_system_2}) and taking into account the solution for $\z$, we get a pair of conditions
\begin{subequations}\label{eq:theorem_system_B_2}
\begin{equation}
      \f^T S_B \f
    + 2 \y^T S_1 \f
    + \y^T S_0 \y
    + 2 \f^T G B A^{-1} \w'
    + 2 \w'^T \y
    = d',
\end{equation}
\begin{equation}
    S_4 \y + S_3 \f - \f' = 0.
\end{equation}
\end{subequations}
Our current aim is to find the form of the matrices $S_0$, $S_1$, $S_3$, $S_4$ and $S_B$  that turns equations (\ref{eq:theorem_system_B_2}) into an identity for all $\y$ satisfying (\ref{eq:solution_A_2}). Since (\ref{eq:solution_A_2}) is an even equation with respect to $\y$, Eqs. (\ref{eq:theorem_system_B_2}) can be split by the parity with respect to $\y$:
\begin{subequations}\label{eq:theorem_system_B_3}
\begin{equation}\label{eq:theorem_system_B_3_1}
      \y^T S_0 \y
    = d' - \f^T S_B \f - 2 \f^T G B A^{-1} \w',
\end{equation}
\begin{equation}\label{eq:theorem_system_B_3_3}
    \y^T \left(S_1 \f + \w'\right)= 0,
\end{equation}
\begin{equation}\label{eq:theorem_system_B_3_2}
    \f' = S_3 \f,
\end{equation}
\begin{equation}\label{eq:theorem_system_B_3_4}
    S_4 \y = 0.
\end{equation}
\end{subequations}
The matrix $S_3$ can be arbitrary. 
By the condition of the theorem, the set of all $\y$ satisfying condition (\ref{eq:solution_A_2}) span the entire $\text{ker} B$. Thus, Eqs. (\ref{eq:theorem_system_B_3_3}), (\ref{eq:theorem_system_B_3_4}) can be rewritten as
\begin{equation}
    S_4 \Pi = 0,\qquad
    \Pi^T (\w' + S_1 \f) = 0.
\end{equation}
Making use of conditions (\ref{eq:S_conditions}), we get
\begin{equation}
    S_4 = 0,\qquad
    (\mathds{1} - B^T G B A^{-1}) \w' + S_1 \f = 0.
\end{equation}
The remaining condition is related with $\text{ker}B$ only:
\begin{equation}\label{eq:theorem_condition_5_1}
       \y^T S_0 \y
    = d' - 2 \f^T G B A^{-1} \w' - \f^T S_B \f,
\end{equation}
for all $\y$ satisfying (\ref{eq:theorem_condition_5_2}). The problem (\ref{eq:theorem_condition_5_1}) is solved in the simpler version of the theorem. According to theorem \ref{th:pre_theorem}, we have
\begin{equation}
    S_0 = \lambda A_\text{pr},\qquad
    d' - 2 \f^T G B A^{-1} \w' - \f^T S_B \f = \lambda (d - \f^T G \f),
\end{equation}
where $\lambda$ is an arbitrary scalar. 
The first condition defines the matrix $S_0$, and the second condition can be represented as
\begin{equation}
    S_B = S_2 + \lambda G,
\end{equation}
where $S_2: U_\f \mapsto U_\f$ is an arbitrary linear operator such that $\f^T S_2 \f = d' - \lambda d - 2 \f^T G B A^{-1} \w'$.
Combining all together, the second system must have the following form
\begin{align}
    A' &= \lambda A + S_1 B + (S_1 B)^T+ B^T S_2 B,\\
        B'& = S_3 B, \\  
      \f' &= S_3 \f,
\end{align}
where $\lambda$ is an arbitrary scalar, $\w'$ and $S_{1,2,3}$ are arbitrary matrices satisfying the following conditions
\begin{align}
     \f^T S_2 \f = d' - \lambda d  - 2 \f^T G B A^{-1} \w',\quad
      (\mathds{1} - B^T G B A^{-1}) \w' + S_1 \f = 0, \quad
      B A^{-1} S_1 = 0.
\end{align}

\end{proof}

\begin{theorem}\label{th:solutions_general}
Let $\mathbf{X} = \{\x\}$ be a non-empty set of all solutions for the following quadratic equation with a system of linear constraints
\begin{equation}\label{eq:main_theorem_system_1}
    \x^T A \x + 2 \w^T \x = d,\qquad B \x = \f,
\end{equation}
and similarly $\mathbf{X}' = \{\x'\}$ is a set of all solutions for
\begin{equation}\label{eq:main_theorem_system_2}
    \x^T A' \x + 2 \w'^T \x = d',\qquad B' \x = \f',
\end{equation}
where  $A, A'$ are symmetric non-degenerate $n_A\times n_A$-matrices, $B, B'$ are $n_B \times n_A$-matrices and $n_B$-vector $\f$ is in the image of $B$. If all solutions $\mathbf{X}$ projected to the $\text{ker}B$ with respect to the pseudo-inner product $A$ span the entire $\text{ker}B$, then the following statement is fair. The entire set of solutions for the first system (\ref{eq:theorem_system_1}) is contained in the set of solutions for the second system (\ref{eq:theorem_system_2}), $\mathbf{X} \subseteq \mathbf{X}'$, if and only if the primed quantities for the second system can be represented as 
\begin{subequations} \label{eq:main_theorem_decomposition}
\begin{align}
    A' &= \lambda A+ S_1 B + (S_1 B)^T+ B^T S_2 B,\\
        B'& = S_3 B, \\  
      \f' &= S_3 \f,
\end{align}
\end{subequations}
where $\lambda$ is an arbitrary scalar, matrices $S_1, S_2, S_3$ are arbitrary matrices constrained by the following condition
\begin{align}\label{eq:main_theorem_conditions}
    &
     d' + \left( 2 \w'^T - \w^T A^{-1} A' \right) A^{-1} \w
     - \lambda (d + \w^T A^{-1} \w)
     = \\\nonumber & =
     (\f + B A^{-1} \w)^T S_2 (\f + B A^{-1} \w) + 2 (\f + B A^{-1} \w)^T G B A^{-1} (\w' - A' A^{-1} \w )
    \\\nonumber &
    (\mathds{1} - B^T G B A^{-1}) \w' = \left(\lambda + B^T (S_1^T + S_2 B) A^{-1} \right) \w - S_1 \f,
    \\\nonumber &
    B A^{-1} S_1 = 0.
\end{align}
\end{theorem}

\begin{proof}
Since $A$ is non-degenerate, it is always possible to make a transformation translating the system (\ref{eq:main_theorem_system_1}) into (\ref{eq:theorem_system_1}) from theorem \ref{th:theorem}. Such a transformation is
\begin{equation}
    \x \to \x - A^{-1} \w.
\end{equation}
After a careful substitution and identification of the translated quantities with original from theorem \ref{th:theorem}, we immediately obtain Eq. (\ref{eq:main_theorem_decomposition}) with conditions (\ref{eq:main_theorem_conditions}).

\end{proof}

\end{document}